\documentclass[11pt,times]{article}
\usepackage[utf8]{inputenc}
\usepackage{amssymb,amsmath}
\usepackage{algorithm}
\usepackage{algpseudocode}
\usepackage{amsthm}
\usepackage{enumerate}
\usepackage[font={small,it}]{caption}
\usepackage{subcaption}
\allowdisplaybreaks
\usepackage[T1]{fontenc}

\usepackage{float}
\newfloat{algorithm}{t}{lop}

\usepackage{tikz}
\usetikzlibrary{decorations.pathmorphing}
\usetikzlibrary{shapes,fit}
\usetikzlibrary{backgrounds}

\usepackage{bm}
\usepackage[font=small,skip=5pt]{caption}
\newtheorem{theorem}{Theorem}[section]

\newtheorem{definition}[theorem]{Definition}
\newtheorem{lemma}[theorem]{Lemma}

\newtheorem{observation}[theorem]{Observation}
\newtheorem{conjecture}{Conjecture}

\newcommand{\Xomit}[1]{}
\newcommand{\vone}{\vspace{.1in}}
\newcommand{\vh}{\vspace{.05in}}

\newcommand{\MinWeightDelta}{Min-Wt-$\Delta$}
\newcommand{\leqsp}{\leq^{\pmb{sprs}}}
\newcommand{\lesssimsp}{\lesssim^{\pmb{sprs}}}
\newcommand{\equivsp}{\equiv^{\pmb{sprs}}}
\newcommand{\Wmax}{$M$ }
\newcommand{\Wmin}{$m$ }

\newcommand{\APSP}{APSD}
\newcommand{\APSPI}{APSP}

\renewcommand\labelenumi{(\roman{enumi})}
\renewcommand\theenumi\labelenumi

\algtext*{EndWhile}
\algtext*{EndFor}
\algtext*{EndIf}

\advance \topmargin by -\headheight
\advance \topmargin by -\headsep
\evensidemargin \oddsidemargin
\begin{document}

\title{Fine-Grained Complexity for Sparse Graphs}
\author{Udit Agarwal $^{\star}$ and Vijaya Ramachandran\thanks{Dept. of Computer Science, University of Texas, Austin TX 78712. Email: {\tt udit@cs.utexas.edu, vlr@cs.utexas.edu}. This work was supported in part by NSF Grant CCF-1320675. The first author's research was also supported in part by a Calhoun Fellowship.}}

\maketitle

\begin{abstract}
We consider the fine-grained complexity of sparse graph problems that currently have $\tilde{O}(mn)$
 time algorithms, where $m$ is the number of edges
and $n$ is the number of vertices in the input graph. This class includes several important path problems on both directed and undirected graphs, including
APSP, MWC (minimum weight cycle), and Eccentricities, which is the problem of computing, for each vertex in the graph, the length of a longest shortest path starting at that vertex.

\vone
We introduce the notion of a sparse reduction which preserves the sparsity of graphs, and we
present near linear-time sparse reductions between
various pairs of graph problems in  the $\tilde{O}(mn)$ class.
Surprisingly, very few of the known nontrivial reductions between problems in the $\tilde{O}(mn)$ class are sparse reductions. 
 In the directed case, our results give a 
partial order on 
a large collection of problems in
the $\tilde{O}(mn)$ class (along with some equivalences), and many of our reductions are very nontrivial.  
In the undirected case we give two nontrivial sparse reductions: from MWC to APSP,
and from unweighted ANSC (all nodes shortest cycles) to a natural variant of APSP. The latter reduction also 
gives an improved algorithm for ANSC (for dense graphs).

\vone
We propose the {\it MWC Conjecture}, a new conditional hardness conjecture that the weight of
a minimum weight cycle in a directed graph cannot be computed in 
time polynomially smaller than $mn$. 
Our sparse reductions for directed
path problems in the $\tilde{O}(mn)$ class 
establish
that several problems in this class, including 2-SiSP (second simple shortest path), $s$-$t$ Replacement Paths, Radius, and Eccentricities,
are MWCC hard.
We also identify Eccentricities as a key problem in the $\tilde{O}(mn)$ class which is simultaneously  MWCC-hard, SETH-hard
 and $k$-DSH-hard, where SETH is the Strong Exponential
Time Hypothesis, and $k$-DSH is the 
hypothesis that a dominating set of size k cannot be computed in time polynomially smaller than $n^k$.

\vone
Our framework using sparse reductions 
 is very
relevant to real-world graphs, which tend to be sparse and for which the  $\tilde{O}(mn)$ time algorithms are the ones  typically
used in practice, and not the $\tilde{O}(n^3)$ time algorithms.
\end{abstract} 

\newpage
\section{Introduction}

In recent years there has been considerable interest in determining the fine-grained complexity of problems 
in P,  see e.g.~\cite{Williams15}.
 For instance, the 3SUM~\cite{GO95} and OV (Orthogonal Vectors)~\cite{Williams05,BI15} problems have been central to the fine-grained complexity of
several problems with quadratic time algorithms, in computational geometry and related areas for 3SUM and in edit distance and related areas
for OV.
APSP (all pairs shortest paths)  has been central to the fine-grained complexity of several path problems with cubic time algorithms
on dense  graphs~\cite{WW10}.
3SUM has a quadratic time algorithm but no sub-quadratic (i.e., $O(n^{2-\epsilon})$ for some constant $\epsilon>0$) time algorithm is known.
It has been shown that a sub-quadratic time algorithm for any of a large number of problems would imply a sub-quadratic time  algorithm
 for  3SUM~\cite{GO95, Williams05}.
 In a similar vein no sub-quadratic time algorithm is known for OV, and  it has been shown that 
 a sub-quadratic time algorithm for LCS or Edit Dis	tance would imply a sub-quadratic time algorithm
for finding 
orthogonal vectors (OV)~\cite{BI15}. 

For several graph problems related to shortest paths that currently have $\tilde{O}(n^3)$
\footnote{$\tilde{O}$ and $\tilde{\Theta}$ 
can hide sub-polynomial factors; 
in our new results they only hide polylog factors.}
time algorithms,  
{\it equivalence} under sub-cubic reductions
has been shown
in work starting with~\cite{WW10}:
 between all pairs shortest paths (APSP) in either a
directed or undirected weighted graph,  finding a second simple shortest path 
from a given vertex $u$ to a given vertex $v$ (2-SiSP) in a weighted directed graph, finding a minimum weight
cycle  (MWC) in a directed or undirected graph, finding a minimum weight triangle (Min-Wt-$\Delta$),
 and a host of other 
problems.
This gives compelling evidence that a large class of 
problems on dense graphs
is unlikely
to have  sub-cubic algorithms as a function of $n$, the number of vertices, unless fundamentally new
algorithmic techniques are developed. 

The {\it Strong Exponential Time Hypothesis (SETH)}~\cite{IP01} states that for every $\delta < 1$, there is a $k$ such that $k$-SAT 
cannot be solved in $O(2^{\delta\cdot n})$ time.
It has been shown that a sub-quadratic time algorithm for 
OV
would
 falsify SETH~\cite{Williams05}. No hardness results relative to SETH are known for either 3SUM or cubic APSP, and the latter problem, in fact, does not
have a SETH hardness result 
for deterministic algorithms
unless NSETH, a nondeterministic version of SETH, is falsified~\cite{CGI+16}.
Other hardness conjectures have been proposed, for instance, for $k$ Clique~\cite{ABW15} and for $k$ Dominating Set ($k$-DSH)~\cite{PW10}.

In this paper we consider a central collection of graph problems related to APSP,
which refines the subcubic equivalence class. 
We let 
$n$ be the number of vertices and $m$ the number of edges.
All of the sub-cubic equivalent
 graph problems mentioned above (and  several others)
 have $\tilde{O}(mn)$ time
algorithms; additionally,
many sub-cubic equivalent problems related to
minimum
 triangle detection and 
 triangle
 listing 
have lower
$O(m^{3/2})$
 time complexities for sparse graphs~\cite{IR78}.
(Checking whether a graph contains a triangle has an even faster $O(m^{1.41})$ time algorithm~\cite{AYZ97} but can also
be computed in sub-cubic time using fast matrix multiplication.) For APSP with arbitrary edge weights, 
 there is 
an $O(mn +n^2 \log\log n)$ time algorithm for directed
 graphs~\cite{Pettie04} and an even faster $O(mn \log \alpha (m,n))$ time algorithm for undirected graphs~\cite{PR05}, 
 where $\alpha$ is a
certain natural inverse of the Ackermann's function. 
(For integer weights, 
there is an $O(mn)$ time algorithm 
 for undirected graphs~\cite{Thorup97} and
an $O(mn +n^2 \log\log n)$ time algorithm for directed graphs~\cite{Hagerup00}.)
When a graph is truly sparse with $m=O(n)$ these bounds are essentially optimal or very close to optimal,
since the size of the output for APSP is $n^2$. Thus, a cubic in $n$ bound for APSP does not fully capture what is currently
achievable by known algorithms, especially since graphs that arise in practice tend to have $m$ close to linear in $n$
or at least are {\it sparse}, i.e., have  $m=O(n^{1 + \delta})$ for $\delta < 1$. 
This motivates our study of the fine-grained complexity of graph path problems
that currently have $\tilde{O}(mn)$ time algorithms.

Another fundamental problem in the $\tilde{O}(mn)$ class is MWC (Minimum Weight Cycle). In both directed and undirected graphs, MWC can be computed
in $\tilde{O}(mn)$ time using an algorithm for APSP.
Very recently Orlin and Sedeno-Noda gave an improved $O(mn)$ time algorithm~\cite{OS17} for directed MWC.
This is an 
important result but the bound still remains $\Omega (mn)$.
Finding an MWC algorithm that runs faster than $mn$ time is a long-standing
open problem in graph algorithms.

Fine-grained reductions and hardness
 results with respect to time bounds that consider 
 only $n$ or only $m$,
 such as bounds of the form 
 $m^2$, $n^2$, $n^c$, $m^{1+\delta}$, 
 are given
  in~\cite{RW13, AWW16, AWY15, AW14,HKNS15}.
 One exception is 
 in~\cite{AGW15} where some  reductions are given for problems 
 with $\tilde{O}(mn)$ time algorithms, such as diameter and some betweenness centrality problems,
 that preserve graph sparsity.
 Several other related results on fine-grained complexity are in~\cite{HKNS15,WW13,Patrascu10, KPP16, JV16,ACLL14, AW14, AWY15, AWW16, Williams15}.
 
 In this paper we present both fine-grained reductions and hardness results 
 for graph problems
 with $\tilde{O}(mn)$ time algorithms, 
 most
 of which are equivalent under sub-cubic reductions
 on dense graphs,
  but now
  taking sparseness of edges into consideration. We use the current long-standing upper bound
 of $\tilde{O}(mn)$ for these problems as our reference, both for our fine-grained reductions and for our
hardness results.
Our results give a 
partial order on  
hardness of several 
problems
in the $\tilde{O}(mn)$ class, 
with equivalence within some subsets of problems,
and a new hardness conjecture (MWCC) for this class.

Our results appear to the first that consider a hardness class with respect to two natural parameters in the input (the $\tilde{O}(mn)$ class in our case). Further, the
 $\tilde{O}(mn)$ time bound has endured for a long time for a large collection of important problems, and hence merits our detailed study.

\section{Our Contributions}  \label{sec:contr}

We will deal with either an unweighted graph $G=(V,E)$ or a weighted graph
$G = (V,E,w)$, where
the weight function is $w: E \rightarrow \mathcal{R^+}$. 
We assume that the vertices have distinct labels with $\lceil \log n \rceil$ bits.
Let \Wmax  and \Wmin  denote the largest and the smallest edge weight, and let  the
{\it edge weight ratio} be $\rho=$ \Wmax/\Wmin.
Let $d_G (x,y)$ denote the length (or weight) of a shortest path from $x$ to $y$ in $G$, and
 for a cycle $C$ in $G$, let $d_C (x,y)$ denote the length of the shortest path from $x$ to $y$ in $C$.
 We deal with the APSP 
 problem\footnote{\label{note1}In an earlier version of this write-up, we used $\rm{APSP'}$  for APSP, and APSP for APSD, respectively.} 
 whose output is the shortest path weights
for all pairs of vertices, 
together with a
concise
 representation of the shortest paths, which in our case is an $n\times n$ matrix, $Last_G$, 
that contains, 
 in position $(x,y)$,
 the predecessor vertex of $y$ on a shortest path from 
$x$ to $y$.
 We also consider the APSD (All Pairs Shortest Distances) problem~\cite{GM97}
 \footnotemark[2]
 which only involves computing the 
 weights of the shortest paths.
  Most of the currently known APSD algorithms, including matrix multiplication based methods for small integer weights~\cite{Seidel95,SZ99,Zwick02}, can compute APSP  in the same bound as APSD.
 We deal with only simple graphs in this paper.

\subparagraph*{I. \underline{Sparse Reductions and the $mn$ Partial Order}.} 

In Definition~\ref{def:mn-red} below, we define the notion of a sparsity preserving reduction (or {\it sparse
reduction} for short) from a graph
problem $P$ to a graph problem $Q$ that allows $P$ to inherit $Q$'s time bound for the graph problem as
a function of both $m$ and $n$, as long as the reduction is efficient. Our definition is in the spirit of a Karp 
reduction~\cite{Karp72}, but slightly more general, since we allow a constant number of calls to $Q$ instead of
just one call in a Karp reduction (and we allow polylog calls for $\tilde{O}(\cdot)$ bounds). 

One could consider
a more general notion of a sparsity preserving reduction in the spirit of Turing reductions as 
in~\cite{GIKW17} (which considers functions of a single variable $n$). However, for all of the
many sparsity preserving reductions we present here, the simpler notion defined below suffices. It should also
be noted that the simple and elegant definition of a Karp reduction suffices for the vast majority of 
known NP-completeness reductions. 
The key difference between our definition and other definitions of fine-grained reductions is 
that it is fine-grained with regard to both $m$ and $n$, and respects the dependence on {\it both parameters.}

It would interesting to see if some of the open problems left by our work on
fine-grained reductions for the $mn$ class can be solved by moving to a more general sparsity preserving reduction in the spirit of a Turing reduction applied to functions of both $m$ and $n$. We 
do not
consider this more general
version here since we do not need it for our reductions.

 \begin{table*}
\small
\centering
\begin{tabular}{| c | l | l |}
\hline
{\sc Reduction} & \hfill {\sc Prior Results (Undirected)}  \hfill $~$ & \hfill  \textbf{\sc New} {\sc Results (Undirected)} \hfill $~$ \\
\hline
\hline
MWC & $\tilde{O}(n^2)$ reduction~\cite{RW11}& \textbf{Sparse} $\bm{\tilde{O}(n^2)}$ \textbf{Reduction} \\
$\leq$ \APSP & \hspace{0.05in} \textit{a.} \underline{goes through Min-Wt-$\Delta$}  & \hspace{0.05in} \textbf{\textit{a.} no intermediate problem} \\
&  \hspace{0.05in} \textit{b.} \underline{$\Theta(n^2)$ edges} in reduced graphs & \hspace{0.05in} \textbf{\textit{b.}} $\bm{\Theta(m)}$ \textbf{edges}\textbf{ in reduced graphs}\\
\hline
ANSC & Sparse \underline{$\tilde{O}(mn^{\frac{3-\omega}{2}})$ reduction}~\cite{Yuster11} & \textbf{Sparse} $\bm{\tilde{O}(n^2)}$ \textbf{Reduction} \\
(Unweighted) & \hspace{0.05in} \textit{a.} \underline{randomized} & \hspace{0.05in} \textbf{\textit{a.} deterministic} \\ 
\APSPI & \hspace{0.05in} \textit{b.} \underline{polynomial calls} to \APSPI & \hspace{0.05in} \textbf{\textit{b.}} $\bm{\tilde{O}(1)}$ \textbf{calls to} \APSPI \\
 & \hspace{0.05in} \textit{c.} gives \underline{randomized $\tilde{O}(n^{\frac{\omega + 3}{2}})$} & \hspace{0.05in} \textbf{\textit{c.} gives deterministic} $\bm{\tilde{O}(n^{\omega})}$ \\
 & \hspace{0.05in} time algorithm~\cite{Yuster11}  & \hspace{0.05in} \textbf{time algorithm} \\
\hline
\end{tabular}
\caption{Our sparse reduction results for undirected graphs. These results are in Sections~\ref{sec:undir} and \ref{sec:otherundir}. Note that Min-Wt-$\Delta$ can be solved in $m^{3/2}$ time.} \label{table1}
\end{table*}

\begin{definition}[\textbf{Sparsity Preserving Graph Reductions}]	\label{def:mn-red}
Given graph problems $P$ and $Q$, 
there is a \emph{sparsity preserving $f(m,n)$  reduction from $P$ to $Q$}, 
denoted by $P \leqsp_{f(m,n)} Q$, if given 
an algorithm for $Q$ that runs in $T_Q (m,n)$ time on graphs with $n$ vertices and $m$ edges,  we can solve $P$ in 
$O(T_Q (m,n)+f(m,n))$ time on graphs with $n$ vertices and $m$ edges,
by making a constant number of oracle calls to $Q$.
\end{definition}

For simplicity, we will refer to a sparsity preserving graph reduction as a {\it sparse reduction}, and we will
say that $P$ {\it sparse reduces to} $Q$ .
Similar to Definition~\ref{def:mn-red},
we will say that $P$ {\it tilde-$f(m,n)$ sparse reduces to $Q$}, denoted by 
$P \lesssimsp_{f(m,n)} Q$,  if, given an algorithm for $Q$ that runs in $T_Q (m,n)$ time,  we can solve $P$ in
$\tilde{O}(T_Q (m,n)+f(m,n))$ time
(by making polylog oracle calls to $Q$ on graphs with $\tilde{O}(n)$ vertices and $\tilde{O}(m)$ edges).
 We will also use $\equivsp_{f(m,n)}$ and
 $\cong^{\pmb{sprs}}_{f(m,n)}$ in place of $\leqsp_{f(m,n)} $ and $\lesssimsp_{f(m,n)}$ when there are reductions
 in both directions.
In a weighted graph
 we allow the $\tilde{O}$ term to have a 
$\log \rho$ factor.
(Recall that 
$\rho = M/m$.)

We present several sparse reductions for problems that currently
have $\tilde{O}(mn)$ time algorithms. This gives rise to a 
partial order on problems that are known to be sub-cubic equivalent, and currently have
 $\tilde{O}(mn)$ time algorithms.  
   For the
 most part, our reductions take $\tilde{O}(m+n)$ time 
 (many are in fact $O(m+n)$ time), except reductions to APSP take $\tilde{O}(n^2)$ time. This ensures
 that any improvement in the time bound for
 the target problem will give rise to the same improvement to the source problem, to within a polylog factor.
 Surprisingly, very few of the known sub-cubic
 reductions carry over to the sparse case due to one or both of the following features. 
 
 \vspace{.05in}
 1.  A central technique used in many of these earlier reductions has been to
 reduce 
 to or from a suitable triangle finding problem. As noted above, in the sparse setting, all
  triangle finding and enumeration problems can be computed in  
$\tilde{O}(m^{3/2})$ time,
 which is
 an asymptotically smaller bound than $mn$ when the graph is sparse.
 
2. Many of the known sub-cubic
 reductions convert a sparse graph into
 a dense one, and
 all known subcubic reductions from a problem in the $mn$ class to a triangle finding problem create
 a dense graph. If any such reduction had been sparse, it would have given an $O(m^{3/2})$ time algorithm
 for a problem whose current fastest algorithm is in the $mn$ class, a major improvement.

 \vspace{.05In}
We present a suite of new sparse
 reductions for the $\tilde{O}(mn)$ time
 class.  Many of our reductions are quite intricate, and for some of our reductions we introduce a new technique of 
{\it bit-sampling} (previously called `bit-fixing' but re-named here to avoid confusion with an un-named
technique used in~\cite{AGW15}).
The full definitions of the problems we consider are in the Appendix.
  Tables~\ref{table1} and \ref{table2} summarize the improvements our reductions achieve over prior results.
 We now give some highlights of our results.

 \begin{table*}
\small
\centering
\begin{tabular}{| c | l | l |}
\hline
{\sc Reduction} & \hfill {\sc Prior Results (Directed)}  \hfill $~$ & \hfill  \textbf{\sc New} {\sc Results (Directed)} \hfill $~$ \\
\hline
\hline
MWC & \underline{$\tilde{O}(n^2)$ reduction} ~\cite{RW11,WW10} & \textbf{Sparse} $\bm{O(m)}$ \textbf{Reduction} \\
$\leq$ 2-SiSP & \hspace{0.05in} \textit{a.} \underline{goes through Min-Wt-$\Delta$} & \hspace{0.05in} \textbf{\textit{a.} no intermediate problem} \\
& \hspace{0.05in} \textit{b.} \underline{$\Theta(n^2)$ edges} in reduced graphs  & \hspace{0.05in} \textbf{\textit{b.}} $\bm{O(m)}$ \textbf{edges in reduced graph}  \\
\hline
2-SiSP & \underline{$\tilde{O}(n^2)$ reductions} ~\cite{GL09,WW10,AGW15} & \textbf{Sparse} $\bm{\tilde{O}(m)}$ \textbf{Reduction} \\
$\leq$ Radius; & \hspace{0.05in} \textit{a.} \underline{goes through Min-Wt-$\Delta$}  & \hspace{0.05in} \textbf{\textit{a.} no intermediate problem} \\
$\leq$ BC & \hspace{0.05in} and a host of other problems  &  \\
& \hspace{0.05in} \textit{b.}  \underline{$\Theta(n^2)$ edges} in reduced graphs  & \hspace{0.05in} \textbf{\textit{b.}} $\bm{\tilde{O}(m)}$ \textbf{edges in reduced graph} \\
\hline
Replacement &  \underline{$\tilde{O}(n^2)$ reductions}~\cite{GL09,WW10,AGW15} & \textbf{Sparse} $\bm{O(m)}$ ($\bm{\tilde{O}(m)}$) \textbf{Reduction}\\
paths & &  \textbf{to ANSC (Eccentricities)} \\ 
$\leq$ ANSC; & \hspace{0.05in} \textit{a.}  \underline{goes through Min-Wt-$\Delta$} & \hspace{0.05in} \textbf{\textit{a.} no intermediate problem} \\
$\leq$ Eccentricities & \hspace{0.05in} and a host of other problems & \hspace{0.05in} \textbf{\textit{b.}} $\bm{O(m)}$ \textbf{edges for ANSC \&} \\ 
&  \hspace{0.05in} \textit{b.}  \underline{$\Theta(n^2)$ edges} in reduced graphs & \hspace{0.05in} $\bm{\tilde{O}(m)}$ \textbf{edges for Eccentricities in}  \\
 & & \hspace{0.05in} \textbf{reduced graph} \\
\hline
ANSC &  \underline{$\tilde{O}(n^2)$ reduction}~\cite{WW10,AGW15} & \textbf{Sparse} $\bm{\tilde{O}(m)}$ \textbf{Reduction} \\
$\leq$ ANBC & \hspace{0.05in} \textit{a.}  \underline{goes through Min-Wt-$\Delta$} & \hspace{0.05in} \textbf{\textit{a.} no intermediate problem} \\
 & \hspace{0.05in} and a host of other problems &  \\ 
&  \hspace{0.05in} \textit{b.}  \underline{$\Theta(n^2)$ edges} in reduced graphs & \hspace{0.05in} \textbf{\textit{b.}} $\bm{\tilde{O}(m)}$ \textbf{edges in reduced graphs}\\
\hline 
\end{tabular}
\caption{Our sparse reduction results for directed graphs. These results are in Sections~\ref{sec:dir}, \ref{sec:otherdir} and  \ref{sec:centrality}. The definitions for these problems are in Appendix~\ref{sec:definitions}. Note that Min-Wt-$\Delta$ can be solved in $m^{3/2}$ time.} \label{table2}
\end{table*}

\begin{figure}
\scriptsize
	\centering
				\begin{tikzpicture}[scale=0.5, every node/.style={circle, draw, inner sep=0pt, minimum width=5pt, align=center}]
            	\node[ellipse, minimum size=0.5cm] (MWC) at (-1,0) {MWC};
            	\node[ellipse, minimum size=0.5cm] (2SiSP) at (3.5,0) {2SiSP};         
            	\node[ellipse, minimum size=0.5cm] (2SiSC) at (3.5,3) {2SiSC};                     	   	
            	\node[ellipse, minimum size=0.5cm] (Replacement) at (10.5,0) {s-t replacement \\ paths};
            	\node[ellipse, minimum size=0.5cm] (ANSC) at (10.5,3) {ANSC};
            	\node[ellipse, minimum size=0.5cm] (Radius) at (10.5,-2.5) {Radius};
            	\node[ellipse, minimum size=0.5cm] (Eccentricities) at (18.5,0) {Eccentricities};
            	\node[ellipse, minimum size=0.5cm] (APSP) at (24.5,0) {APSP};
            	\draw[->, line width=1pt] (MWC) -- (2SiSP) {};
            	\begin{scope}[->,every node/.style={ right}]
            	\draw[->, line width=1pt] (2SiSP) to  [bend right=60] (2SiSC) {};
            	\node (l) at (3,1.5) {\cite{AR16}};
            	\draw[->, line width=1pt] (2SiSC) to  [bend right=60] (2SiSP) {};
            	\end{scope}
            	\draw (2,-1) -- (5,-1) -- (5,4) -- (2, 4) -- cycle {};    
            	\draw[->, line width=1pt, dashed] (2SiSP) -- (Replacement) {};
            	\draw[->, line width=1pt] (ANSC) to  [bend right=30]  (Replacement) {};
            	\draw[->, line width=1pt] (Replacement) to  [bend right=30]  (ANSC) {};
            	\draw (7,-1) -- (13.5,-1) -- (13.5,4) -- (7, 4) -- cycle {};                	
            	\draw[->, line width=1pt, decorate, decoration={zigzag}] (2SiSP) -- (Radius) {};
            	\draw[->, line width=1pt, dashed] (Radius) -- (Eccentricities) {};
            	\draw[->, line width=1pt, decorate, decoration={zigzag}] (Replacement) -- (Eccentricities) {};
            	\begin{scope}[->,every node/.style={ right}]
            	\draw[->, line width=1pt, dashed] (Eccentricities) -- (APSP) node[midway, label=above:$n^2$] {};
            	\end{scope}
            	\end{tikzpicture}
            \caption{Our sparse reductions for weighted directed graphs (reductions related to centrality problems are in Figure~\ref{fig:reductionsBC}). The regular edges represent sparse $O(m+n)$ reductions, the squiggly edges represent tilde-sparse $O(m+n)$ reductions, and the dashed edges represent reductions that are trivial. The $n^2$ label on dashed edge to APSP denotes an $O(n^2)$ time reduction.}

    \label{fig:reductions}
\end{figure}
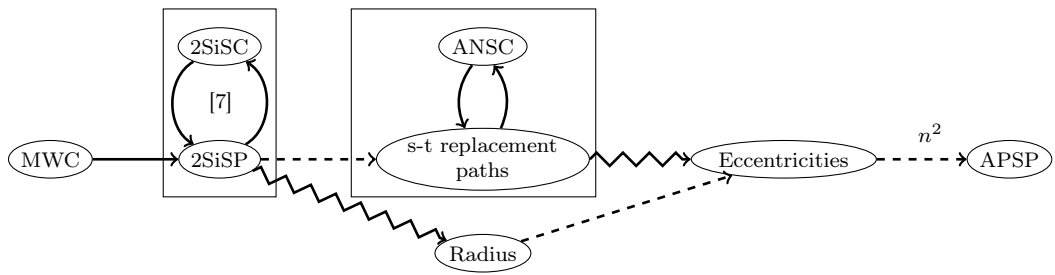

\subparagraph*{(a) Sparse Reductions for Undirected Graphs:}

Finding the weight of a minimum weight cycle (MWC) is a fundamental problem.  
A simple  sparse $O(m+n)$ reduction from MWC to APSD is
known for directed graph 
but it
does not work in
the undirected case
mainly because an edge can be traversed in either direction in an undirected graph, and known 
algorithms for the directed case would create non-simple paths when applied to an undirected graph.
 Roditty and Williams~\cite{RW11}, in a follow-up 
 to~\cite{WW10},
 pointed out the challenges of reducing from undirected MWC to 
 APSD in sub-$n^{\omega}$ time,
 where $\omega$ is the matrix multiplication exponent, and then gave
a $\tilde{O}(n^2)$
reduction from undirected MWC to undirected Min-Wt-$\Delta$ in a 
{\it dense} bipartite 
 graph. But a reduction 
 that increases the density of the graph is not helpful 
 in our sparse setting. Instead, in this paper we give a sparse 
$\tilde{O}(n^2)$
 time reduction from undirected MWC to APSD.
 Similar techniques allow us to obtain a
 sparse
 $\tilde{O}(n^2)$
  time reduction from undirected ANSC (All Nodes Shortest Cycles)~\cite{Yuster11,SW16}, which asks for a shortest
cycle through every vertex) to APSP.
This reduction improves the running time for 
unweighted ANSC in {\it dense} graphs~\cite{Yuster11}, since we can now solve it
in $\tilde{O}(n^{\omega})$ time using the 
unweighted APSP algorithm in~\cite{Seidel95,AGM+92}. 
  Our ANSC reduction and resulting improved algorithm is only for
 unweighted graphs and extending it to weighted graphs appears to be challenging.
 
   We introduce
 a new 
 {\it bit-sampling technique} \footnote{In an earlier version, we called this bit-fixing but that term was often confused with a `bit-encoding' technique used in~\cite{AGW15}.}
 in these reductions.
 This technique contains
   a simple 
 construction with exactly $\log n$ hash functions for Color Coding~\cite{AYZ97} 
 with 2 colors (described in detail in Section~\ref{sec:undir}).
 Our bit-sampling method also gives 
 the
 first near-linear time algorithm for $k$-SiSC 
  in weighted
 undirected graphs.
$k$-SiSC is the cycle variant of $k$-SiSP~\cite{KIM82}, and our reduction that gives a fast algorithm for $k$-SiSC in weighted undirected
graphs is given in
Appendix~\ref{sec:kSiSC}.

 Section~\ref{sec:undir} 
 summarizes the proof of Theorem~\ref{thm1} below, and Section~\ref{sec:otherundir}  
 proves Theorems~\ref{thm1} and \ref{thm2} in full.

\begin{theorem} \label{thm1}
 In a weighted undirected $n$-node $m$-edge graph with edge weight ratio $\rho$,
MWC  can be computed with  $2 \cdot \log n \cdot \log \rho$ 
 calls to APSD 
on graphs with $2n$ nodes, at most $2m$ edges, and edge weight ratio at most $\rho$,
with $O(n+m)$ cost for constructing
each reduced graph, and with additional $O(n^2 \cdot \log n \cdot \log (n\rho))$ processing time.
Additionally, 
edge weights are preserved, and
every edge in the reduced graph retains its 
corresponding
edge weight from the original graph.
 Hence, MWC $\lesssimsp_{n^2}$ APSD.
\end{theorem}

In undirected graphs with integer weights at most $M$, APSD 
can be computed in
$\tilde{O}(M \cdot n^{\omega})$ time~\cite{Seidel95,SZ99}. 
In~\cite{RW11}, the authors give an $\tilde{O}(M \cdot n^{\omega})$ time algorithm for undirected MWC in such graphs by
preprocessing using a result in~\cite{LL09} and then
making $\tilde{O}(1)$ calls to an APSD 
algorithm.
By applying our sparse reduction in Theorem~\ref{thm1} we can get an alternate {\it simpler}
$\tilde{O}(M \cdot n^{\omega})$ time algorithm for  
MWC in undirected graphs (the sparsity of our reduction is not relevant here except for
the fact that it is also an $\tilde{O}(n^2)$ reduction).

The following result 
gives an improved algorithm for ANSC in undirected unweighted graphs.

\begin{theorem}\label{thm:reductions-algorithms}\label{thm2}
In undirected unweighted graphs,
ANSC $\lesssimsp_{n^2}$  APSP 
and ANSC 
can be computed
 in $\tilde{O}(n^{\omega})$ time.
\end{theorem}

\subparagraph*{(b) Directed Graphs.}  
We give several nontrivial sparse reductions starting from MWC in directed graphs,
 as noted in the following theorem
 (also highlighted in Figure~\ref{fig:reductions}).
 
\begin{theorem}[\textbf{Directed Graphs.}]  \label{thm:weighted-directed}
 In weighted directed graphs:
\begin{enumerate}
\item MWC $\leqsp_{m+n}$ 2-SiSP $\leqsp_{m+n}$ $s$-$t$ replacement paths $\equivsp_{m+n}$ ANSC $\lesssimsp_{m+n}$ Eccentricities
\item 2-SiSP  $\lesssimsp_{m+n}$ Radius $\leqsp_{m+n}$ Eccentricities
\item 2-SiSP $\lesssimsp_{m+n}$ Betweenness Centrality
\end{enumerate}
\end{theorem}

In Section~\ref{sec:dir} we present a brief overview
of our sparse reduction
 from 2-SiSP to Radius for directed graphs. The remaining reductions are in 
 Section~\ref{sec:otherdir},  and in Section~\ref{sec:centrality} we
also present nontrivial sparse reductions from 2-SiSP and ANSC to versions of the {\it betweenness 
centrality} problem, to complement a collection of sparse reductions in~\cite{AGW15} for
betweenness centrality problems.

\subparagraph*{II.  \underline{Conditional Hardness Results}.}

Conditional hardness under fine-grained reductions falls into several categories: 
For example, 3SUM hardness~\cite{GO95} holds for several problems to which 3SUM reduces in sub-quadratic time,
and OV-hardness holds for several problems to which OV has sub-quadratic reductions. 
 By a known reduction from  SETH  to OV~\cite{Williams05} OV-hardness implies SETH-hardness as well.
 The $n^3$ equivalence class for path problems in dense graphs~\cite{WW10}
gives sub-cubic hardness for APSP in dense graphs and for the other problems in this class.

In this paper 
our focus is on
the $mn$ class, the class of graph path problems for which the
current best algorithms run in $\tilde{O}(mn)$ time.  This class differs from all previous classes considered
for fine-grained complexity since it depends on two parameters of the input, $m$ and $n$.
 To formalize hardness results for this class we
first make precise the notion of a sub-$mn$ time bound, which we capture with the following
definition.

\begin{definition}[\textbf{Sub-$mn$}] \label{def:sub-mn} 
A function $g(m,n)$ is \emph{sub-$mn$} if 
$g(m,n) = O(m^{\alpha} \cdot n^{\beta})$, where $\alpha, \beta$ are constants such that $\alpha + \beta <2$.
\end{definition}

A straightforward application of the notions of sub-cubic and sub-quadratic to the two-variable function $mn$ would have resulted in a
simpler but less powerful definition for sub-$mn$, namely that
 of requiring a time bound $O((mn)^{\delta})$, for some $\delta<1$. Another weaker form of
the above definition would have been to require $\alpha\leq 1$ and $\beta \leq 1$ with at least one of the two being strictly less than
1. The above definition is more general than either of these. It
considers a bound of the form $m^3/n^2$ to be sub-$mn$ even though such a
bound for the graph problem with be larger than $n^3$ for dense graphs. Thus, it is a very strong
definition of sub-$mn$ when applied to hardness results. (Such a definition could be abused when
giving sub-$mn$ reductions, but as noted in part I above, all of our sub-$mn$ reductions are 
linear or near-linear in the sizes of input and output, and thus readily satisfy the sub-$mn$ definition 
while being very efficient.)

 Based on our fine-grained reductions for directed graphs, we propose the
following conjecture:

\begin{conjecture}~\label{conj:mwc}
\emph{(MWCC: Directed Min-Wt-Cycle Conjecture)}
There is no sub-$mn$ time algorithm for MWC (Minimum Weight Cycle) in directed graphs. 
\end{conjecture}

Directed MWC is a natural candidate for hardness for the $mn$ class since it is a fundamental
problem for which a simple $\tilde{O}(mn)$ time algorithm has been known for many decades,
and very recently, an $O(mn)$ time
algorithm~\cite{OS17}.
But a sub-$mn$ time algorithm remains as elusive as ever. Further, through the
fine-grained reductions that we present in this paper, many other problems in the $mn$ class 
have MWCC hardness for sub-$mn$ time as noted in the following theorem.

\begin{theorem}\label{thm:MWC-hard}
Under MWCC, the following problems on directed graphs do not have sub-$mn$ time algorithms:
2-SiSP, 2-SiSC, $s$-$t$ Replacement Paths, ANSC, Radius, Betweenness Centrality, and
Eccentricities. 
\end{theorem}

The problems in the above theorem
are a subset of the problems that are sub-cubic equivalent to directed MWC,
and hence one could also strengthen the MWCC conjecture to state that directed MWC
has neither a
sub-$mn$ time algorithm nor
a subcubic (in $n$) time algorithm. As noted above, these two classes of time bounds are incomparable
as neither one is contained in the other.

\vh
\noindent
{\bf SETH and Related Problems.}
Recall that SETH conjectures
 that for every $\delta < 1$ there exists a $k$ such that there is no `$2^{\delta\cdot n}$' time algorithm for $k$-SAT. 
No SETH-based conditional lower bound is known for either 3SUM or for dense APSP.
In fact, 
it is
{\it unlikely} that dense APSP would have a hardness result for 
 deterministic sub-cubic
algorithms, relative to SETH,
since this would falsify NSETH~\cite{CGI+16}.

Despite the above-mentioned negative result in~\cite{CGI+16} for SETH hardness for 
sub-cubic equivalent problems, we now observe
that SETH-hardness does hold for a key 
MWCC-hard
problem in the $mn$ class. In particular  we observe that
a SETH hardness construction, used in~\cite{RW13} to obtain
sub-$m^2$ hardness for Eccentricities in graphs with $m=O(n)$, can be used to
establish the following result.

\begin{theorem}[\textbf{Sub-$mn$ Hardness Under SETH}] \label{thm:seth-hardness}
Under SETH, Eccentricities does not have a sub-$mn$ time algorithm in an unweighted or weighted graph,
either directed or undirected.
\end{theorem}

The  
{\it $k$ Dominating Set Hypothesis ($k$-DSH)}~\cite{PW10}
states that there exists $k_0$ ($k_0=3$ in~\cite{PW10})
such that for all $k \geq k_0$, a dominating set of size $k$
in an undirected graph on $n$ vertices cannot be found in $O(n^{k - \epsilon})$ for any
constant $\epsilon >0$. This conjecture formalizes a long-standing open problem, and 
it was shown in~\cite{PW10} that falsifying $k$-DSH
would falsify SETH. 

For even values of $k$, it was shown in~\cite{RW13} that solving Eccentricities  in 
$O(m^{2-\epsilon})$ time, for
any constant $\epsilon>0$, would falsify $k$-DSH. We extend this result
to all values of $k$, both odd and even, and establish it for the sub-$mn$ class (which includes
$O(m^{2-\epsilon})$), to obtain the following theorem.

\begin{theorem}[\textbf{Conditional Hardness Under $k$-DSH}] \label{thm:hardness}
Under $k$-DSH, 
Eccentricities does not have a sub-$mn$ time algorithm in an unweighted or weighted graph,
either directed or undirected. More precisely, if Eccentricities can be solved in 
$O(m^{\alpha} n^{2-\alpha - \epsilon})$ time, then for any $k \geq 3 + (2 \alpha/\epsilon)$,
a dominating set of size $k$  can be found in $O(n^{k - \epsilon})$ time.
\end{theorem}

We give the proof of  
Theorem~\ref{thm:hardness}
in Section~\ref{sec:seth}, and this also establishes
 Theorem~\ref{thm:seth-hardness},
since hardness under $k$-DSH implies hardness under SETH. 

\vh
\noindent
{\bf Eccentricities as a Central Problem for $mn$.} We observe that directed Eccentricities is a central 
problem
in the $mn$ class: If a sub-$mn$ time
algorithm is obtained for directed Eccentricities, not only would it refute $k$-DSH, SETH and MWCC
(Conjecture~\ref{conj:mwc}), it would also imply sub-$mn$ time algorithms for 
several MWCC-hard problems:
2-SiSP, 2-SiSC, $s$-$t$ Replacement Paths, ANSC and Radius
in directed graphs as well as Radius and Eccentricities in undirected graphs.
(The undirected versions of 2-SiSP, 2-SiSC,  and $s$-$t$ Replacement Paths 
 have near-linear time algorithms and are not relevant, 
 and there is no known sparse reduction from undirected ANSC to Eccentricities.)

\vspace{.05in}
\noindent
{\bf APSP.} APSP has a special status in the $mn$ class. Since its output size is $n^2$,
it has near-optimal algorithms~\cite{Thorup97, Hagerup00, Pettie04, PR05} for
graphs with $m=O(n)$. Also, the $n^2$ size for the APSP output means that
any inference made through sparse reductions to APSP will not be based on
a sub-$mn$ time bound but instead on a  sub-$mn + n^2$ time bound. It also turns out that
the SETH and $k$-DSH hardness results for Eccentricities depend crucially on staying with a purely
sub-$mn$ bound, and hence even though Eccentricities has a simple sparse $n^2$ reduction to APSP,
we do not have SETH or $k$-DSH hardness for computing APSP in sub-$mn + n^2$ time.

\vh
\noindent
{\bf Betweenness Centrality (BC).} We discuss this problem in Section~\ref{sec:centrality}.
Sparse reductions for several variants of BC were given in~\cite{AGW15} but none established
MWCC-hardness. In Section~\ref{sec:centrality} we give nontrivial sparse reductions to establish
MWCC hardness for some important variants of BC.

\vh
\noindent
{\bf Time Bounds for Sparse Graphs and their Separation under SETH}. 
It is readily seen
that the $\tilde{O}(m^{3/2})$ bound for triangle finding problems is a
better bound than the $\tilde{O}(mn)$ bound for the $mn$ class. But imposing
a total
ordering
on functions
of two variables requires some care. 
 For example, 
 maximal 2-connected subgraphs of a given directed graph
 can be computed in $O(m^{3/2})$ time~\cite{CHI+17}
 as well as in $O(n^2)$ time~\cite{HKL15}.
 Here, $m^{3/2}$ is a better bound for very sparse graphs, and $n^2$ for very dense graphs.
In Section~\ref{sec:sparse-time-bound} we 
give natural definition
of
what it means for 
for one
time bound to be smaller than another time bound for sparse graphs.
By our definitions in
Section~\ref{sec:sparse-time-bound}, $m^{3/2}$ is a smaller
  time bound than both $mn$ and $n^2$
  for sparse graphs.
Our definitions in Section~\ref{sec:sparse-time-bound}, in conjunction with Theorem~\ref{thm:seth-hardness}, establish that the
problems related to triangle listing must have {\it provably smaller}  time bounds
 for sparse graphs than Eccentricities under $k$-DSH or SETH.

\vh
\noindent
{\bf Sub-cubic equivalence for $\tilde{O}(n^3)$ versus the sub-$mn$ Partial Order for $\tilde{O}(mn)$.}
 Our sparse reductions
give a 
 partial order 
 on hardness
 for several graph problems with $\tilde{O}(mn)$ time bounds, and a hardness
conjecture for directed graphs relative to a specific problem MWC. 
While this partial order may appear to be weaker than  the sub-cubic equivalence class of problems 
with $\tilde{O}(n^3)$ time bound for $n$-node graphs~\cite{WW10},  we also
 show that under SETH or $k$-DSH there is a {\it provable separation}
 between Eccentricities, a problem in the $\tilde{O}(mn)$ partial order, and the $O(m^{3/2})$ class of triangle finding problems, even though all of these problems 
are equivalent under sub-cubic reductions. Thus, if we assume SETH, the equivalences achieved under sub-cubic reductions
{\it cannot hold} for the sparse versions of the problems (i.e., when parameterized by both $m$ and $n$).

\vh
The results for conditional hardness relative to SETH and $k$-DSH are fairly straightforward,
 and they adapt earlier hardness results to our framework.
 The significance of these hardness results is in the new
insights they give into our inability to make improvements to some long-standing
time bounds for important problems on sparse graphs. On the other hand, many of
our sparse reductions are highly intricate, and overall these reductions give
a  partial order 
(with several equivalences)   
 on the large class of graph problems that currently have
$\tilde{O}(mn)$ time algorithms.

\vh
\noindent
{\bf Roadmap.}
The rest of the paper is organized as follows. Sections~\ref{sec:undir} and \ref{sec:dir} present key
examples of our sparse reductions for undirected and directed graphs, with full details in
Sections~\ref{sec:otherundir} and \ref{sec:otherdir}. In Sections~\ref{sec:seth} and \ref{sec:sparse-time-bound}
we present SETH and $k$-DSH hardness results, and the resulting
provable split of the sub-cubic equivalence class under these hardness
results for sparse time bounds. Section~\ref{sec:centrality} discusses
betweenness centrality. The Appendix gives the definitions of the problems we consider.

\section{Weighted Undirected Graphs: MWC $\lesssimsp_{n^2}$ APSD}\label{sec:undir}

In undirected graphs, 
 the only known sub-cubic reduction from MWC to 
 APSD~\cite{RW11} 
 uses
 a dense
reduction to Min-Wt-$\Delta$.
Described
in~\cite{RW11} for integer edge weights
of value at most $M$,
it
first 
uses an algorithm in~\cite{LL09} to compute, in
$O(n^2 \cdot \log n \log nM)$
time, a 2-approximation $W$ to the weight of a
minimum weight cycle as well as shortest paths between all pairs of vertices with pathlength at most 
 $W/2$.
The reduced graph for Min-Wt-$\Delta$ 
is constructed as a (dense) bipartite graph with edges to represent all of these shortest paths, 
together with the edges of the original graph in one side of the bipartition. 
This results in each triangle in the reduced graph corresponding to a cycle in the original graph, and with  a minimum weight cycle guaranteed to be present as a triangle.
An MWC is then constructed using 
a simple explicit construction for Color Coding~\cite{RW11} with  2 colors.

The approach in~\cite{RW11} does not work in our case, as we are dealing with sparse reductions.
Instead, we give a sparse reduction directly from MWC to 
APSD.
In contrast to~\cite{RW11}, where finding a minimum weight 3-edge triangle gives the MWC in the original graph, 
in our reduction the MWC is constructed as a path $P$ in 
a
reduced graph followed by a shortest path in the original graph.

One may ask if we can sparsify the dense reduction from MWC to Min-Wt-$\Delta$ 
in~\cite{RW11}, but such a
reduction, though very desirable, would immediately refute MWCC  and would achieve a major
breakthrough by giving an $\tilde{O}(m^{3/2})$ time algorithm for undirected MWC.

 We now sketch our sparse reduction from 
undirected
 MWC 
 to 
 APSP.
 (In Section~\ref{sec:otherundir} we refine this to a sparse reduction to 
 APSD.)
 It is well known that in any cycle $C = \langle v_1, v_2, \ldots, v_l \rangle$  in a weighted undirected graph
$G=(V,E,w)$ there exists an edge $(v_i,v_{i+1})$ on $C$
 such that
$\lceil \frac{w(C)}{2} \rceil - w(v_i,v_{i+1}) \leq d_C (v_1,v_i) \leq \lfloor \frac{w(C)}{2} \rfloor$ and 
$\lceil \frac{w(C)}{2} \rceil - w(v_i,v_{i+1}) \leq d_C (v_{i+1},v_1) \leq \lfloor \frac{w(C)}{2} \rfloor$.
The above edge, $(v_i,v_{i+1})$, is called the \emph{critical edge} of $C$
 with respect to
 the start vertex $v_1$ in~\cite{RW11}.
 
We will make use of the following simple observation 
proved in Section~\ref{sec:otherundir}.

\begin{observation}	\label{obs:MinWtCyc}
Let $G = (V,E,w)$ be a weighted undirected graph.
 Let $C = \langle v_1, v_2, \ldots, v_l \rangle$ be a minimum weight cycle in $G$, and let
$(v_p,v_{p+1})$ be its critical edge with respect to $v_1$.
WLOG assume that $d_G (v_1,v_p) \geq d_G (v_1,v_{p+1})$.
If $G'$ is obtained by removing edge $(v_{p-1},v_p)$ from $G$, then the path $P = \langle v_1, v_l, \ldots, v_{p+1}, v_p \rangle$ is a shortest path from $v_1$ to
$v_p$ in $G'$.
\end{observation}

 In our reduction we construct
a collection of graphs $G_{i,j,k}$, each with $2n$ vertices (containing 2 copies of $V$) and $O(m)$ edges,
with the guarantee that, for the minimum weight cycle $C$,
in at least one of the graphs
the edge $(v_{p-1},v_p)$ (in Observation~\ref{obs:MinWtCyc}) will not connect across the two copies of $V$ and the path $P$ of
Observation~\ref{obs:MinWtCyc}
will be present. Then, if a call to 
 APSP
computes $P$  as a shortest path from $v_1$ to $v_p$ (across the two copies of $V$), 
we can verify that edge 
$(v_p,v_{p+1})$ is not the last edge on the computed shortest path from $v_1$ to $v_p$ in $G$,  and so we can form the
concatenation of these two paths as a possible candidate for a minimum weight cycle. 
The challenge is to construct a small collection of graphs where we can ensure that
the  path we identify in one of the
derived graphs is in fact the simple path $P$ in the input graph. 

Each $G_{i,j,k}$ has two
 copies of each vertex $u \in V$, $u^1 \in V_1$ and $u^2 \in V_2$. All edges in $G$ are present 
 on the vertex set $V_1$, but
 there is no edge that connects any pair of vertices within $V_2$. 
 In $G_{i,j,k}$ there is an edge from $u^1 \in V_1$ to $v^2 \in V_2$ iff there is an edge from $u$ to
$v$ in $G$,
 and $u$'s $i$-th bit is $j$, and $\frac{M}{2^k} < w(u,v) \leq \frac{M}{2^{k-1}}$.
All the edges in $G_{i,j,k}$ retain their weights from $G$.
Thus, the edge $(u^1,v^2)$ is present in $G_{i,j,k}$ with weight $w$  only if $(u,v)$ is an edge in $G$ with the same weight $w$ and further, certain conditions (as described above)
hold for the indices $i,j,k$.
Here, 
$1 \leq i \leq \lceil \log n \rceil$, $j \in \{0,1\}$ and $k \in \{1, 2, \ldots, \lceil \log \rho \rceil \}$, so
we have 
$2 \cdot \log n \cdot \log \rho$
graphs. 
Figure \ref{fig:b} depicts the construction of graph $G_{i,j,k}$.

The first condition for an edge $(u^1,v^2)$ to be present in $G_{i,j,k}$ is that $u$'s $i$-th bit must be $j$.
This ensures that there exist a graph where the edge $(v^1_{p-1},v^2_p)$ is absent and the edge $(v^1_{p+1},v^2_p)$ is present (as 
$v_{p-1}$ and $v_{p+1}$ differ on at least 1 bit).
To contrast with a similar step in~\cite{RW11}, we need to find the path $P$ in the sparse derived graph, while in~\cite{RW11} it suffices to look for the 2-edge path
  that represents $P$ in a triangle in their dense reduced graph. 

The second condition --- that an edge $(u^1,v^2)$ is present only if $\frac{M}{2^k} < w(u,v) \leq \frac{M}{2^{k-1}}$ ---
 ensures that there is a graph $G_{i,j,k}$ in which, not only is  edge $(v^1_{p+1},v^2_p)$ is present  and
edge $(v^1_{p-1},v^2_p)$ is absent as noted by the first condition, but also the shortest path from 
$v^1_1$ to $v^2_p$ is in fact the path $P$ in Observation~\ref{obs:MinWtCyc}, and does not correspond to a false path 
where an edge in $G$ is 
traversed twice.
In particular, we show that this second condition allows us to exclude a shortest path from $v^1_1$ to 
$v^2_p$ of the following form: take the shortest path from $v_1$ to $v_p$ in $G$ on vertices in $V_1$,
then take an edge $(v^1_p,x^1)$, and then the edge $(x^1,v^2_p)$. Such a path, which has 
weight $d_C(v_1, v_p) + 2 w(x, v_p)$,  could be shorter 
than the desired path, which has weight  $d_C (v_1,v_{p+1}) + w(v_{p+1},v_p)$.  In our reduction we
avoid selecting this ineligible path by requiring that the weight of the selected path should not exceed 
$d_G(v_1, v_p)$ by more than $M/ 2^{k-1}$. We show 
 that these conditions suffice
to ensure that $P$ is identified in one of the $G_{i,j,k}$, and no spurious path of shorter length is
identified.
Notice that, in contrast to~\cite{RW11}, we do not estimate the MWC weight by computing a
2-approximation. Instead, this second condition allows us to identify the critical edge in the appropriate
graph.

\begin{figure}
\scriptsize
    \centering
        \makebox[.3\textwidth]{
            \begin{tikzpicture}
			\node[ellipse, draw, minimum width = 1cm, minimum height = 3cm](V1) [label=below:$V_1$] at (-3,0) {};
			\node[ellipse, draw, minimum width = 1cm, minimum height = 3cm](V2) [label=below:$V_2$] at (3,0) {};
			\draw[->] (-4,-1) -- (-3,0);
			\node[text width = 0.5cm] at (-4.5,-1) {all edges from E};
			\draw[->] (4,-1) -- (3,0);
			\node[text width = 1.5cm] at (4.5,-1) {no edges between vertices in $V_2$};
			\node[circle,draw,inner sep=0pt] (u1)[label=above:$u^1$] at (-3,1) {}; 
			\node[circle,draw,inner sep=0pt] (v2)[label=above:$v^2$] at (3,1) {}; 
			\draw (u1) -- (v2);
			\node[circle,draw,inner sep=0pt] (a1)[label=above:$a^1$] at (-3,0.5) {}; 
			\node[text width = 0.5cm] at (-4.25,0.75) {$w(u,a)$};
			\draw (u1) -- (a1);
			\draw[->] (-3.5,0.75) -- (-3,0.75);
			\node[circle,draw,inner sep=0pt] (a11) at (3,0.5) {}; 
			\draw (a1) -- (a11);
			\node[circle,draw,inner sep=0pt] (a12) at (3,0.75) {}; 
			\draw (a1) -- (a12);
			\node[circle,draw,inner sep=0pt] (a13) at (3,-0.4) {}; 
			\draw (a1) -- (a13);
			\node[circle,draw,inner sep=0pt] (c1)[label=above:$c^1$] at (-3,-0.5) {}; 
			\node[circle,draw,inner sep=0pt] (c11) at (3,-0.5) {}; 
			\draw (c1) -- (c11);	
			\node[circle,draw,inner sep=0pt] (c12) at (3,-0.25) {}; 
			\draw (c1) -- (c12);	
			\node[circle,draw,inner sep=0pt] (c13) at (3,0.35) {}; 
			\draw (c1) -- (c13);	
			\node[circle,draw,inner sep=0pt] (f1)[label=above:$f^1$] at (-3,-1) {}; 
			\node[circle,draw,inner sep=0pt] (f11) at (3,-1) {};
			\draw (f1) -- (f11); 
			\node[circle,draw,inner sep=0pt] (f12) at (3,-0.75) {};
			\draw (f1) -- (f12); 
			\node[circle,draw,inner sep=0pt] (f13)[label=right:$g^2$] at (3,-1.25) {};
			\draw (f1) -- (f13); 
			\node[text width = 1.5cm] at (0,-1.5) {$w(f,g)$};
			\node[text width = 4cm] at (0,1.5) {$(u^1,v^2)$ present if $(u,v) \in E$ \\ and $u$'s $i$-th bit is $j$ \\ and $\frac{M}{2^k} < w(u,v) \leq \frac{M}{2^{k-1}}$};
            \end{tikzpicture}}
    \caption{Construction of $G_{i,j,k}$.}
    \label{fig:b}
\end{figure}
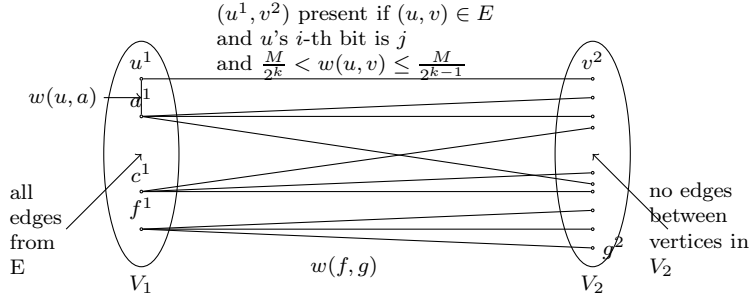

The following lemma establishes the correctness of the resulting sparse reduction described in Algorithm 
 MWC-to-APSP.
 The proof and full details are in Section~\ref{sec:otherundir},
where Figure~\ref{fig:reductionsMWC} illustrates how Observation~\ref{obs:MinWtCyc} applies to the MWC in a suitable $G_{i,j,k}$.

\begin{lemma}	\label{lemma:MinWtCycAPSP}
Let $C = \langle v_1, v_2, \ldots, v_{l} \rangle$ be a minimum weight cycle in $G$ and let $(v_p, v_{p+1})$ be its critical edge with
respect to the start vertex $v_1$.
Assume $d_G (v_1, v_p) \geq d_G (v_1, v_{p+1})$.
Then there exist $i \in \{1,\ldots,\lceil \log n \rceil\}$, $j \in \{0,1\}$
and $k \in \{1, 2, \ldots, \lceil \log \rho \rceil \}$ such that the following conditions hold:\\
(i) $d_{G_{i,j,k}} (v^1_1,v^2_{p}) + d_G (v_1, v_p) = w(C) $;\\
 (ii) $Last_{G_{i,j,k}} (v^1_1,v^2_{p}) \neq Last_G (v_1, v_p)$;\\
(iii) $d_{G_{i,j,k}} (v^1_1,v^2_{p}) \leq d_G (v_1, v_p) + \frac{M}{2^{k-1}}$.

The converse also holds: If there exist vertices $y, z$ in $G$ with the above three properties satisfied for $y=v_1$,
$z=v_p$ for one of the graphs $G_{i,j,k}$ using a weight $wt$ in place of $w(C)$ in part 1, then there exists a cycle in $G$ that passes through  $z$ of weight at most $wt$.
\end{lemma}

\begin{algorithm}[H]	
\scriptsize
\caption*{ MWC-to-APSP($G$)}
\begin{algorithmic}[1]	
\State $wt \leftarrow \infty$
\For{$1\leq i\leq \lceil \log n\rceil$, $j \in \{0,1\}$, and $1\leq k\leq \lceil \log \rho \rceil$}
       \State{Compute APSP on $G_{i,j,k}$}
	\For{$y,z \in V$}
		\State {\bf  if} $d_{G_{i,j,k}} (y^1, z^2) \leq d_G (y,z) + \frac{M}{2^{k-1}}$ {\bf then} check if $Last_{G_{i,j,k}} (y^1, z^2) \neq Last_G (y, z)$ 	\label{alg1:check1}
		\State {\bf if} both checks in Step~\ref{alg1:check1}  hold {\bf then} $wt \leftarrow min(wt, d_{G_{i,j,k}} (y^1,z^2) + d_G (y, z))$ 
	\EndFor
\EndFor
\State {\bf return} $wt$
\end{algorithmic} 
\end{algorithm}

\subparagraph*{Bit-sampling and Color Coding:}
Color Coding 
is a method 
introduced 
by
Alon, Yuster and Zwick ~\cite{AYZ97}.
For the special case of 2 colors,
the method constructs 
a collection $C$ of $O(\log n)$ different 2-colorings on an $n$-element set $V$, such that for every 
pair $\{x,y\}$ in $V$, there
 is a 2-coloring in $C$ that assigns different colors to $x$ and $y$.
 When the elements of $V$ have unique $\log n$-bit labels, e.g., by numbering them from 0 to $n-1$,
our bit-sampling 
method on index $i$ (ignoring indices $j$ and $k$)  can be viewed as an explicit construction
of exactly 
$\lceil \log n\rceil$ hash functions for the 2-perfect hash family: the $i$-th hash function  
assigns to each element
the $i$-th bit in its label as its color. 

In our construction we actually use $2\log n$ functions (using both $i$ and $j$) since we need a stronger version of color coding  where,
for any pair of vertices $x$, $y$, there is a hash function that assigns color 0 to $x$ and 1 to $y$ and another that assigns 1 to $x$ and 0 to $y$.
This is needed in order to ensure that when $x=v_{p-1}$ and $y=v_{p+1}$, the edge $(v^1_{p-1},v^2_p)$ is absent and the edge
$(v^1_{p+1},v^2_p)$ is present.
 A different variant of Color Coding  with 2 colors 
 is used in~\cite{RW11} in their dense reduction from 
undirected MWC to Min-Wt-$\Delta$,
and we do not immediately see how to
apply our 
bit-sampling
technique 
there.

\Xomit{
Our bit-fixing method differs from
a `bit-encoding' technique used in some reductions in~\cite{AGW15}, 
which
creates $2 \log n$ new vertices
 with $O(\log n)$ bit labels and adds edges to these vertices
based on their bit patterns, in order
to maintain sparsity
while preserving paths from the original graph.
(We use this bit-encoding technique from~\cite{AGW15} in our sparse reduction from 
2-SiSP to Radius in Section~\ref{sec:dir}.)
However, the bit-fixing technique we use in our reduction
here is different, and it selectively samples the edges from the original graph to be placed in the reduced graph.
}
Our bit-sampling method differs from a `bit-encoding' technique used in some reductions 
in~\cite{AGW15, ACK16}, where the objective is to
 preserve sparsity in the constructed graph while 
 also
 preserving paths from the original graph
$G=(V,E)$. This technique
creates paths between 
two copies of $V$ by adding $\Theta(\log n)$ new
vertices with 
$O(\log n)$ bit labels,
 and using the $O(\log n)$ bit labels on these new vertices to induce the desired paths in the constructed graph.
 The bit-encoding technique (from~\cite{AGW15}) is useful for certain types of reductions, and we use it
in our sparse reduction from 2-SiSP to Radius in Section~\ref{sec:dir},
 and from 2-SiSP to BC in Section~\ref{sec:centrality}.

 However, the bit-sampling technique we use in our reduction here is different.
 Here the objective is to selectively sample the edges from the original graph 
 to be placed in the reduced graph, based on the bit-pattern of the end points and the edge weight.
 In our construction, we create $\Theta(\log n)$ 
different
 graphs, where
in each graph
 the copies of $V$ are connected by single-edge paths, without requiring 
 additional intermediate vertices.
\section{Weighted Directed Graphs: 2-SiSP $\lesssimsp_{m+n}$ Radius }\label{sec:dir}

A sparse $O(n^2)$ reduction  from 2-SiSP to APSP was given in~\cite{GL09}.
Our sparse reduction from  2-SiSP to Radius refines this result and the sub-$mn$ partial order by plugging the Radius 
and Eccentricities problems
within the sparse reduction chain from 2-SiSP to APSP.     
 Also, in Section~\ref{sec:otherdir} we show MWC $\leqsp_{m+n}$ 2-SiSP,
 thus establishing MWC-hardness for both 2-SiSP and Radius.
Our 2-SiSP to Radius reduction here is unrelated to the 
sparse
 reduction in~\cite{GL09}
 from 2-SiSP to APSP.

A sub-cubic reduction from Min-Wt-$\Delta$ to Radius is in~\cite{AGW15}.
This reduction transforms the problem of finding a minimum weight triangle to the problem of computing Radius by creating a 4-partite graph.
However, in order
to use this reduction to 
show 
MWC
hardness of the Radius problem
we 
would
need to show that Min-Wt-$\Delta$ is MWC-hard.
But such a reduction would give an $O(m^{3/2})$ time algorithm for MWC, thus refuting the MWC Conjecture.
Instead, we
present a more complex reduction from 2-SiSP, which
is MWC-hard (as shown in Section~\ref{sec:otherdir}).

 The input is $G = (V,E,w)$,  with source  $s$  and sink $t$ in $V$, and a shortest path $P$ ($s=v_0\rightarrow v_1\rightsquigarrow v_{l-1}\rightarrow v_l=t$). We need to compute a second simple  $s$-$t$
shortest path.

 Figure~\ref{fig:2SiSPRadius} gives an example of our 
 reduction to an input $G''$ to the Radius problem
for $l = 3$.
In $G''$
we first map every edge $(v_j,v_{j+1})$ lying on $P$ to the vertices $z_{j_o}$ and 
$z_{j_i}$ such that the shortest path from $z_{j_o}$ to $z_{j_i}$ corresponds to the shortest path from $s$ to $t$ avoiding the edge
$(v_j,v_{j+1})$.
We then add vertices $y_{j_o}$ and $y_{j_i}$ in the graph and connect them to vertices $z_{j_o}$ and $z_{j_i}$ 
 by adding edges $(y_{j_o}, z_{j_o})$ and $(y_{j_i}, z_{j_i})$, and then additional edges from $y_{j_o}$ to other $y_{k_o}$ and $y_{k_i}$ vertices 
such that the longest
shortest path from $y_{j_o}$ is to the vertex $y_{j_i}$, which in turn corresponds to the shortest path from $z_{j_o}$ to $z_{j_i}$.
In order to preserve sparsity,
 we have an  interconnection from each  $y_{j_o}$ vertex to all $y_{k_i}$ vertices 
(except for $k = j$)
 with a sparse construction
 by using  
 $2\log n$
 additional vertices $C_{r,s}$ in a manner
 similar to a bit-encoding technique used in~\cite{AGW15} in 
 their
 reduction from \MinWeightDelta{}  to Betweenness Centrality 
 (this technique however,  
 is different from the  new `bit-fixing' technique used 
 in Section~\ref{sec:undir}),
and we have two additional vertices $A,B$ with suitable edges to
induce connectivity among the 
$y_{j_o}$
vertices. 
In our construction, we ensure that the center is 
one of the $y_{j_o}$ vertices and hence computing the Radius in the reduced
graph gives the minimum among all the shortest paths from $z_{j_o}$ to $z_{j_i}$.
This corresponds to a shortest replacement path from $s$ to $t$.
Further details are in Section~\ref{sec:otherdir},  where the following three claims are proved.

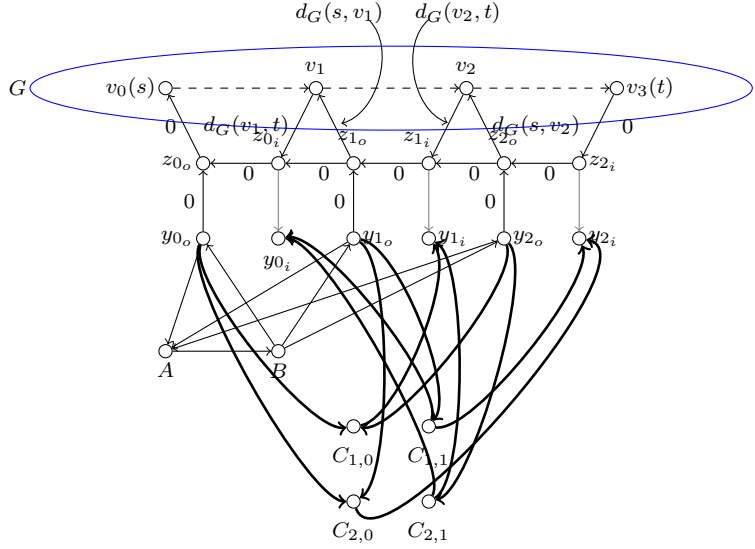
\begin{figure}
\scriptsize
    \centering
        \makebox[.3\textwidth]{
            \begin{tikzpicture}[scale=0.5, every node/.style={circle, draw, inner sep=0pt, minimum width=5pt}]
            \node (v0)[label=left:$v_0(s)$] at (0,0)  {};
            \node (v1)[label=above:$v_1$] at (4,0) {};
            \node (v2)[label=above:$v_2$]	at (8,0)	{};
            \node (v3)[label=right:$v_3(t)$] at (12,0)	{};
            \node[ellipse, color=blue, fit=(v0)(v1)(v2)(v3), inner sep=3mm](G) [label=left:$G$] {};
            \node (z0o)[label=left:$z_{0_o}$] at (1,-2)  {};
            \node (z0i) at (3,-2)  {};
            \node (z1o)[label=above:$z_{1_o}$] at (5,-2)  {};
            \node (z1i) at (7,-2)  {};
            \node (z2o)[label=above:$z_{2_o}$] at (9,-2)  {};
            \node (z2i)[label=right:$z_{2_i}$] at (11,-2)  {};
            \draw[->] (v0) -- (v1) [dashed] {};
            \draw[->] (v1) -- (v2) [dashed] {};
            \draw[->] (v2) -- (v3) [dashed] {};
            \begin{scope}[->,every node/.style={ right}]
            \node (labelz0i)[label=above:$z_{0_i}$] at (2.5,-2) {};
            \node (labelz1i)[label=above:$z_{1_i}$] at (6.5,-2) {};
            \draw[->] (z0o) -- (v0) node[midway, label=left:$0$] {};        
            \draw[->] (z0i) -- (z0o) {};
            \node (labelz0iz0o)[label=below:$0$] at (2,-1.65) {};
            \draw[->] (v1) -- (z0i) node[midway, label=left:${d_G(v_1,t)}$] {};
            \draw[->] (z1o) -- (z0i) {}; 
            \node (labelz1oz0i)[label=below:$0$] at (4,-1.65) {};
            \draw[->] (z1i) -- (z1o) {};
            \node (labelz1iz1o)[label=below:$0$] at (6,-1.65) {};
            \draw[->] (z1o) -- (v1) {};
            \node (labelz1ov1)[label=left:${d_G(s,v_1)}$] at (6,2) {};
            \node (z1ov1start) at (5,2) {};
            \node (z1ov1end) at (4.25,-1) {};
            \draw[->] (z1ov1start) to  [bend left=60, distance=10mm]  (z1ov1end) {};
            \draw[->] (v2) -- (z1i) {};
            \node (labelv2z1i)[label=right:${d_G(v_2,t)}$] at (6,2) {};
            \node (v2z1istart) at (7,2) {};
            \node (v2z1iend) at (7.5,-1) {};
            \draw[->] (v2z1istart) to  [bend right=60, distance=10mm]  (v2z1iend) {};
            \draw[->] (z2o) -- (z1i) {};
            \node (labelz2oz1i)[label=below:$0$] at (8,-1.65) {};
            \draw[->] (z2o) -- (v2) {}; 
            \node (labelz2ov2)[label=right:${d_G(s,v_2)}$] at (8.05,-1) {};
            \draw[->] (v3) -- (z2i) node[midway, label=right:$0$] {}; 
            \draw[->] (z2i) -- (z2o) {};
            \node (labelz2iz2o)[label=below:$0$] at (10,-1.65) {};
            \end{scope}
            \node (y0o)[label=left:$y_{0_o}$] at (1,-4)  {};
            \node (y0i)[label=below:$y_{0_i}$] at (3,-4)  {};
            \node (y1o)[label=right:$y_{1_o}$] at (5,-4)  {};
            \node (y1i)[label=right:$y_{1_i}$] at (7,-4)  {};
            \node (y2o)[label=right:$y_{2_o}$] at (9,-4)  {};
            \node (y2i)[label=right:$y_{2_i}$] at (11,-4)  {};
            \begin{scope}[->,every node/.style={ right}]
            \draw[->] (y0o) -- (z0o) node[midway, label=left:$0$] {};
            \draw[->] (y1o) -- (z1o) node[midway, label=left:$0$] {};
            \draw[->] (y2o) -- (z2o) node[midway, label=left:$0$] {};
            \draw[->, color=gray] (z0i) -- (y0i) {};
            \draw[->, color=gray] (z1i) -- (y1i) {};
            \draw[->, color=gray] (z2i) -- (y2i) {};
            \end{scope}
            \node (A)[label=below:$A$] at (0,-7) {};
            \node (B)[label=below:$B$] at (3,-7) {};
            \begin{scope}[->,every node/.style={ right}]
            \draw[->] (A) -- (B) {};
            \draw[->] (y0o) -- (A) {};
			 \draw[->] (y1o) -- (A) {};  
   			 \draw[->] (B) -- (y0o) {};
			 \draw[->] (B) -- (y1o) {};
			 \draw[->] (y2o) -- (A) {};     
			 \draw[->] (B) -- (y2o) {};
            \end{scope}
            \node (C10)[label=below:$C_{1,0}$] at (5,-9) {};
            \node (C11)[label=below:$C_{1,1}$] at (7,-9) {};
            \node (C20)[label=below:$C_{2,0}$] at (5,-11) {};
            \node (C21)[label=below:$C_{2,1}$] at (7,-11) {};
            \begin{scope}[->,every node/.style={ right}]
            \draw[->, line width=1pt] (y0o) to  [bend right=60, distance=10mm] (C10) {};
            \draw[->, line width=1pt] (y0o) to  [bend right=60, distance=10mm] (C20) {};
            \draw[->, line width=1pt] (C11) to  [bend right=60, distance=10mm] (y0i) {};
            \draw[->, line width=1pt] (C21) to  [bend right=60, distance=10mm] (y0i) {};
            \draw[->, line width=1pt] (y1o) to  [bend left=60, distance=10mm] (C11) {};
            \draw[->, line width=1pt] (y1o) to  [bend left=60, distance=10mm] (C20) {};
            \draw[->, line width=1pt] (C10) to  [bend right=60, distance=10mm] (y1i) {};
            \draw[->, line width=1pt] (C21) to  [bend right=60, distance=10mm] (y1i) {};
            \draw[->, line width=1pt] (y2o) to  [bend left=60, distance=10mm] (C10) {};
            \draw[->, line width=1pt] (y2o) to  [bend left=60, distance=10mm] (C21) {};
            \draw[->, line width=1pt] (C11) to  [bend right=60, distance=10mm] (y2i) {};
            \draw[->, line width=1pt] (C20) to  [bend right=120, distance=20mm] (y2i) {};            
            \end{scope}
            \end{tikzpicture}}
    \caption{$G''$ for $l = 3$. The gray and the bold edges have weight $\frac{11}{9}M'$ and $\frac{1}{3}M'$ respectively. All the outgoing (incoming) edges from (to) $A$ have weight $0$ and the outgoing edges from $B$ have weight $M'$.}
    \label{fig:2SiSPRadius}
\end{figure}

 \textit{(i) For each $0\leq j\leq l-1$, the longest shortest path in $G''$ from $y_{j_o}$ is to the vertex $y_{j_i}$.}\\
\textit{(ii) A shortest path from $z_{j_o}$ to $z_{j_i}$ corresponds to a  replacement path for the edge $(v_j,v_{j+1})$.}\\
\textit{(iii) One of the vertices among $y_{j_o}$'s is a center of $G''$.}

Thus 
after computing
the radius in $G''$,
 we can use $(i)$, $(ii)$, and $(iii)$ to compute the weight of a shortest replacement
path from $s$ to $t$, which is a second simple shortest path from $s$ to $t$.
The cost of this reduction is $O(m + n \log n)$. For full details see Section~\ref{sec:otherdir}.
\section{Conditional Hardness Under $k$-DSH}	\label{sec:seth}

 The following lemma shows that a sub-$mn$ time algorithm for
Diameter in an unweighted graph, either undirected or directed, would falsify 
$k$-DSH. This proves Theorems~\ref{thm:hardness} and Theorem~\ref{thm:seth-hardness}
since the Diameter of a graph can be 
computed in $O(n)$ time after one call to Eccentricities on the same graph.
Diameter is in the $\tilde{O}(mn)$ class, but at this time we do not know if 
Diameter is MWCC-hard.

\begin{lemma}	\label{thm:kDominatingDiameter}
Suppose for some constant $\alpha$  there is an $O(m^{\alpha} \cdot n^{2- \alpha - \epsilon})$
time algorithm, for some $\epsilon >0$,  for solving Diameter in an unweighted $m$-edge $n$-node graph, either undirected or
directed. Then there exists a $k'>0$ such that 
for all $k\geq k'$, the $k$-Dominating Set problem can be solved in $O(n^{k - \epsilon})$ time.
\end{lemma}

\begin{proof}
When $k$ is even we use a construction in~\cite{RW13}.
To determine if  undirected graph $G = (V,E)$ has a
 $k$-dominating set
 we form 
 $G' = (V',E')$, where $V'= V_1 \cup V_2$,  with $V_1$ containing a vertex for each subset of $V$ of size $k/2$  and $V_2=V$. We add an edge from a vertex $v \in V_1$ to a vertex $x \in V_2$ if the subset corresponding to $v$ does not dominate $x$.
We induce a clique in the vertex partition $V_2$. As shown in~\cite{RW13}, $G'$ has diameter 3 if $G$ 
has a dominating set of size $k$ and has diameter 2 otherwise, and this gives the reduction
 when $k$ is even.
 
 If $k$ is odd, so $k = 2r+1$, we 
 make $n$ calls to graphs derived from $G' = (V',E')$ as follows, where now each vertex in $V_1$ represents a subset 
 of $r$ vertices in $V$.
For each $x\in V$ let $V_x$ be the set $\{x\} \cup \{$neighbors of $x$ in $G\}$, and let $G_x$ be the subgraph
of $G'$ induced on $V-V_x$. 
 If $G$ has a dominating set $D$ of size $k$ that includes vertex $x$ then
consider any partition of the remaining $2r$ vertices in $D$ into two subsets of size $r$ each, and let
$u$ and $v$ be the vertices corresponding to these two sets in $V_1$.  Since all paths from $u$ to $v$
in $G_x$ pass through $V_2 -V_x$, there is no path of
length 2 from $u$ to $v$ since every vertex in $V_2 - V_x$ is covered by either $u$ or $v$. Hence
the diameter of $G_x$ is greater than 2 in this case. But if there is no dominating set of size $k$ that includes
$x$ in $G$, then for any $u, v \in V_1$, at least one vertex in $V_2 - V_x$ is not covered by both $u$ and $v$ and hence there is a path of length 2 from $u$ to $v$. If we now compute the diameter in each of  
graphs $G_x$, $x\in V$, we will
detect a graph with diameter greater than 2 if and only if $G$ has a dominating set of size $k$.

Each graph $G_x$ has $N=O(n^r)$ vertices and $M=O(n^{r+1})$ edges.
If we now assume that  Diameter can be computed in time $O(M^{\alpha} \cdot N^{2-\alpha - \epsilon})$, 
then the above algorithm for $k$ Dominating Set runs in time
$O(n \cdot M^{\alpha} \cdot N^{2-\alpha - \epsilon}) = O(n^{2r+1 - \epsilon r + \alpha})$, which is 
$O(n^{k - \epsilon})$ time when $k \geq 3 +  \frac{2\alpha}{\epsilon}$.
  The analysis is similar for $k$ even.
In the directed case, we get the same result by replacing every edge in $G'$ with two directed edges in
opposite directions. 
\end{proof}

\section{Time Bounds for Sparse Graphs}\label{sec:sparse-time-bound}

Let  $T(m,n)$ be a function  which is defined for $m \geq n-1$.
We will interpret $T(m,n)$ as a time for an algorithm on a connected graph and we will refer to $T(m,n)$ as a {\it time bound} for a graph
problem.
We now focus on formalizing the notion of a time bound  $T(m,n)$ being smaller than another time
bound $T'(m,n)$ for sparse graphs. 

If the time bounds $T(m,n)$ and $T'(m,n)$ are of the form $m^{\alpha}n^{\beta}$, then 
one possible way to check if $T(m,n)$ is smaller than $T'(m,n)$ is to check if the exponents of $m$ and $n$ in $T(m,n)$ are individually
smaller than the corresponding exponents in $T'(m,n)$.
But using this approach, we would not be able to compare between time bounds $m^{1/2}n$ and $mn^{1/2}$.
Another possible way is to use a direct extrapolation from the single variable case and define $T(m,n)$ to be (polynomially) smaller than $T'(m,n)$ 
if $T(mn) = O((T'(m,n))^{1-\epsilon})$ for some constant $\epsilon>0$. But such a definition would completely ignore the dependence of
the functions on each of their two variables. 
We would want our definition to take into account the sparsity of the graph, i.e., as a graph becomes sparser, the smaller time bound has
smaller running time.
To incorporate this idea, our definition below asks for
$T(m,n)$ to be a factor of $m^{\epsilon}$ smaller than $T'(m,n)$, for some $\epsilon >0$. Further, this requirement is placed only on
{\it sufficiently sparse graphs} (and for a weakly smaller
 time bound, we also require a certain minimum edge density).
The consequence of this definition is that 
when one time bound is not dominated by the other for all values of $m$, the domination needs to hold for sufficiently sparse graphs in order
for the dominated function to be a smaller 
 time bound for sparse graphs.
 
\begin{definition}[\textbf{Comparing Time Bounds for Sparse Graphs}]\label{def:faster-sparse-time} 
Given two time bounds $T(m,n)$ and $T'(m,n)$,
\vspace{-0.1in}
\begin{enumerate}
\item $T(m,n)$ is a \emph{smaller time bound than} $T'(m,n)$  \emph{for sparse graphs} if there exist constants  $\gamma, \epsilon >0$ such that 
$T(m,n) = O\left(\frac{1}{m^{\epsilon} }\cdot T'(m,n)\right)$
for all values of $m=O(n^{1+\gamma})$.
\item $T(m,n)$ is a \emph{weakly smaller  time bound than} $T'(m,n)$  \emph{for sparse graphs}  if there exists a positive constant  $\gamma$ such that for any constant $\delta$ with $\gamma > \delta >0$,
there exists an $\epsilon > 0$ such that $T(m,n) = O\left(\frac{1}{m^{\epsilon} }\cdot T'(m,n)\right)$ for all values of $m$ in the range $m=O(n^{1+\gamma})$ and $m= \Omega (n^{1+\delta})$.
\end{enumerate}
\end{definition}

Part (i) in above definition  requires a polynomially smaller (in $m$) bound for $T(m,n)$ relative to $T'(m,n)$ for sufficiently sparse graphs.
For example, $\frac{m^3}{n^2}$ is a smaller time bound than
$m^{3/2}$, which in turn is a smaller time bound than $mn$;  
$m^2$ is a smaller time bound than $n^3$.
A time bound of $n \sqrt m$ is a weakly smaller 
bound than $m \sqrt n$
 for sparse graphs 
 by part (ii) but not a smaller
  bound since the two bounds coincide when $m=O(n)$.

Our definition for comparing time bounds for sparse graphs is quite strong as it allows us to compare a wide range of time bounds.
For example, using this definition we can say that $n \sqrt m$ is a weakly smaller bound than $m \sqrt n$ for sparse graphs.
Whereas if we use the possible approaches that we discussed before then we would not be able to compare these two time bounds.

With Definition~\ref{def:faster-sparse-time} in hand,
 the following lemma is straightforward.

\begin{lemma}	\label{lemma:smallerTimeBound}
Let $T_1(m,n) = O(m^{\alpha_1}n^{\beta_1})$ and $T_2(m,n) = O(m^{\alpha_2}n^{\beta_2})$  be two 
time bounds, where $\alpha_1, \beta_1, \alpha_2, \beta_2$ are constants.\\
(i) $T_1(m,n)$ is a smaller time bound than $T_2(m,n)$ for sparse graphs if $\alpha_2 + \beta_2 > \alpha_1 + \beta_1$.\\
(ii) $T_1(m,n)$ is a weakly smaller time bound than $T_2(m,n)$ for sparse graphs if $\alpha_2 + \beta_2 = \alpha_1 + \beta_1$, and $\alpha_2 > \alpha_1$.
\end{lemma}

Definition~\ref{def:faster-sparse-time}, in conjunction with
Lemma~\ref{lemma:smallerTimeBound} and 
 Theorem~\ref{thm:seth-hardness}, lead to
the following 
{\it provable separation}
 of time bounds for sparse graph problems
  in the sub-cubic equivalence class:

\begin{theorem}[\textbf{Split of Time Bounds for Sparse Graphs.}] \label{thm:sparsetime}
 Under either SETH or $k$-DSH,
triangle finding problems in the sub-cubic equivalence class have algorithms with a smaller 
 time bound for sparse graphs
than \emph{any algorithm} we can design for
Eccentricities.
\end{theorem}

\section{ Reduction Details for Undirected Graphs}	\label{sec:otherundir}

In Section ~\ref{sec:undir}, we gave an overview of our sparse reduction from MWC to 
APSP.
Here in Section~\ref{sec:fullMinWtCyctoAPSP}, we provide full details of this reduction,
and refine it to a reduction to 
 APSD.
We then describe a $\tilde{O}(n^2)$ sparse reduction from ANSC to 
APSP
in unweighted undirected graphs in Section~\ref{sec:ANSCtoAPSP}.

We 
start with 
stating
from~\cite{RW11}
the notion of a `critical edge'.

\begin{lemma} [\cite{RW11}]
Let $G = (V,E,w)$ be a weighted undirected graph, where $w: E \rightarrow \mathcal{R}^{+}$, and let
$C = \langle v_1, v_2, \ldots, v_l \rangle$ be a cycle in $G$.
There exists an edge $(v_i,v_{i+1})$ on $C$ such that
$\lceil \frac{w(C)}{2} \rceil - w(v_i,v_{i+1}) \leq d_C (v_1,v_i) \leq \lfloor \frac{w(C)}{2} \rfloor$ and 
$\lceil \frac{w(C)}{2} \rceil - w(v_i,v_{i+1}) \leq d_C (v_{i+1},v_1) \leq \lfloor \frac{w(C)}{2} \rfloor$.\\
The edge $(v_i,v_{i+1})$ is called the \emph{critical edge} of $C$ with respect to the start vertex $v_1$.
\end{lemma}

\subsection{Reducing Minimum Weight Cycle to APSD}	\label{sec:fullMinWtCyctoAPSP}

We first describe a useful property of a minimum weight cycle.

\begin{lemma}	\label{lemma:MinWtCycProp-full}
Let $C$ be a minimum weight cycle in weighted undirected graph $G$.
Let $x$ and $y$ be two vertices lying on the cycle $C$ and let $\pi^{1}_{x,y}$ and $\pi^{2}_{x,y}$ be the paths from $x$ to $y$ in $C$.
WLOG assume that $w(\pi^{1}_{x,y}) \leq w(\pi^{2}_{x,y})$.
Then 
$\pi^{1}_{x,y}$ is a shortest path from $x$ to $y$ and
$\pi^{2}_{x,y}$ is a second simple shortest path from $x$ to $y$, i.e. a path from $x$ to $y$ that is shortest among all paths 
from $x$ to $y$ that are not identical to $\pi^{1}_{x,y}$.
\end{lemma}

\begin{proof}
Assume to the contrary that $\pi^3_{x,y}$ is a second simple shortest path from $x$ to $y$ of weight less than $w(\pi^2_{x,y})$.
Let the path $\pi^3_{x,y}$ deviates from the path $\pi^1_{x,y}$ at some vertex $u$ and then it merges back at some vertex $v$.
Then the subpaths from $u$ to $v$ in $\pi^{1}_{x,y}$ and $\pi^{3}_{x,y}$ together form a cycle of weight strictly less than $w(C)$,
resulting in a contradiction as $C$ is a minimum weight cycle in $G$.
\end{proof}

Observation~\ref{obs:MinWtCyc} follows, since the path $P$ there
must be either a shortest path or a second simple shortest path in
$G$ by the above lemma, so in $G'$ it must be a shortest path.

Consider the graphs $G_{i,j,k}$ as described in Section~\ref{sec:undir}.
In the following two lemmas we identify three key properties of a path  
$\pi$ 
from $y^1$ to $z^2$ ($y \neq z$)
 in a $G_{i,j,k}$ that (I) will be satisfied by the path $P$ in 
 Observation~\ref{obs:MinWtCyc}
 for $y^1=v_1^1$ and $z^2=v_p^2$
in some $G_{i,j,k}$ (Lemma~\ref{lemma:MinWtCycAPSP1}),
and (II) will cause a simple cycle in $G$ to be contained in the
concatenation of $\pi$  with the shortest path from $y$ to $z$
computed by 
APSP 
(Lemma~\ref{lemma:MinWtCycAPSP2}).
Once we have these two Lemmas in hand, it gives us a method to find a minimum weight cycle in $G$ by calling 
 APSP
on each $G_{i,j,k}$ and then
identifying all pairs $y^1, z^2$ in each graph that satisfy these properties. Since the path $P$ is guaranteed to be one of the pairs, and no spurious path will be identified, the minimum weight
cycle can be identified.  We now fill in the details.

\begin{figure}
\scriptsize
	\centering
    	\begin{subfigure}{.4\textwidth}
            	\begin{tikzpicture}[scale=0.5, every node/.style={circle, draw, inner sep=0pt, minimum width=5pt, align=center}]
            	\node[circle,draw,inner sep=0pt] (v1)[label=below:$v_1$] at (0,2) {}; 
            	\node[circle,draw,inner sep=0pt] (vp-1)[label=below:$v_{p-1}$] at (4,0) {}; 
            	\node[circle,draw,inner sep=0pt] (vp)[label=below:$v_p$] at (6,2) {}; 
            	\node[circle,draw,inner sep=0pt] (vp+1)[label=below:$v_{p+1}$] at (4,4) {}; 
            	\draw[line width=1pt, decorate, decoration={zigzag}] (v1) to [bend right=10] (vp-1) {};
            	\draw[line width=1pt] (vp-1) -- (vp) {};
            	\draw (vp) -- (vp+1) {};
            	\draw[decorate, decoration={zigzag}] (v1) to [bend left=10]  (vp+1) {};
            	\begin{scope}[->,every node/.style={ right}]
               \node[text width = 3.5cm] at (8,4) {};
               \node at (-1,4) {$G$} ;
               \end{scope}
            	\end{tikzpicture}
            	\caption{MWC $C$ in $G$. Shortest path \\ from $v_1$ to $v_p$ is highlighted.}
            	\label{fig:MWCCycle}
		\end{subfigure}    
		~
		\begin{subfigure}{.4\textwidth}
			\begin{tikzpicture}[scale=0.5, every node/.style={circle, draw, inner sep=0pt, minimum width=5pt, align=center}]
            	\node[circle,draw,inner sep=0pt] (v1)[label=below:$v^1_1(y^1)$] at (0,2) {}; 
            	\node[circle,draw,inner sep=0pt] (vp-1)[label=below:$v^1_{p-1}$] at (4,0) {}; 
            	\node[circle,draw,inner sep=0pt] (vp)[label=below:$v^2_p(z^2)$] at (6,2) {}; 
            	\node[circle,draw,inner sep=0pt] (vp+1)[label=below:$v^1_{p+1}$] at (4,4) {}; 
            	\draw[decorate, decoration={zigzag}] (v1) to [bend right=10] (vp-1) {};
            	\draw[dashed, line width=1pt] (vp-1) -- (vp) {};
            	\draw[line width=1pt] (vp) -- (vp+1) {};
            	\draw[line width=1pt, decorate, decoration={zigzag}] (v1) to [bend left=10]  (vp+1) {};
            	\begin{scope}[->,every node/.style={ right}]
            \node[text width = 4 cm] at (8,2) {$\bullet$ $i$, $j$, $k$ are such that: \\ 1. $i$-th bit of $v_{p+1}$ is $j$, \\ 2. $i$-th bit of $v_{p-1}$ is $\overline{j}$ and \\ 3. $\frac{M}{2^k} < w(v_p,v_{p+1}) \leq \frac{M}{2^{k-1}}$. \\ \vspace{.05in} $\bullet$ Alg. MWC-to-APSP' will compute $w(C)$ in this $G_{i,j,k}$ in $wt$ (line 6)};
            \node at (-1,4) {$G_{i,j,k}$} ;
            \end{scope}
            	\end{tikzpicture}
            \caption{Shortest path from \\ $y^1$ to $z^2$ in $G_{i,j,k}$.}
   	        \label{fig:MWCGijk}
   	  \end{subfigure}
    \caption{Here Figure (a) represent the MWC $C$ in $G$. The path $\pi_{y,z}$ (in bold) is the shortest path from $y$ to $z$ in $G$. The path $\pi_{y^1,z^2}$ (in bold) in Figure (b) is the shortest path from $y^1$ to $z^2$ in $G_{i,j,k}$: where the edge $(v^1_{p-1},v^2_p)$ is absent due to $i, j$ bits. The paths $\pi_{y,z}$ in $G$ and $\pi_{y^1,z^2}$ in $G_{i,j,k}$ together comprise the MWC $C$.}
    \label{fig:reductionsMWC}
\end{figure}
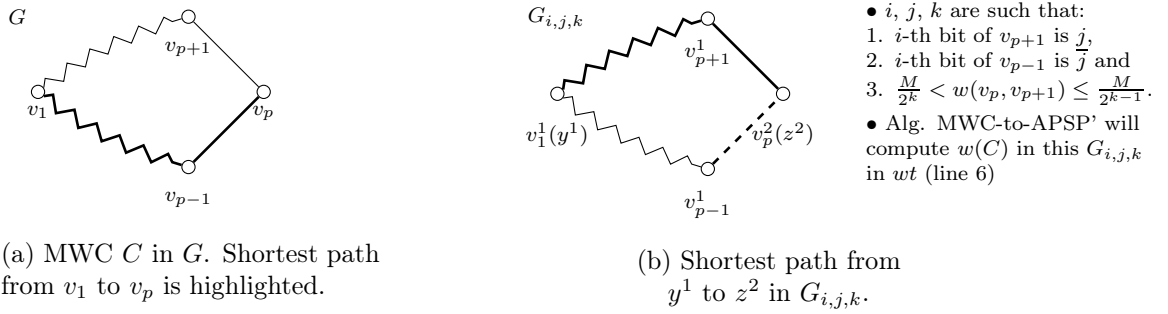
\begin{lemma}	\label{lemma:MinWtCycAPSP1}
Let $C = \langle v_1, v_2, \ldots, v_{l} \rangle$ be a minimum weight cycle in $G$ and let $(v_p, v_{p+1})$ be its critical edge with
respect to the start vertex $v_1$.
WLOG assume that $d_G (v_1, v_p) \geq d_G (v_1, v_{p+1})$.
Then there exists an $i \in \{1,\ldots,\lceil \log n \rceil\}$, $j \in \{0,1\}$
and $k \in \{1, 2, \ldots, \lceil \log \rho \rceil \}$ such that the following conditions hold:
\vspace{-.175in}
\begin{enumerate}
\item $d_{G_{i,j,k}} (v^1_1,v^2_{p}) + d_G (v_1, v_p) = w(C)$
\vspace{-.075in}
\item $Last_{G_{i,j,k}} (v^1_1,v^2_{p}) \neq Last_G (v_1, v_p)$ 
\vspace{-.075in}
\item $d_{G_{i,j,k}} (v^1_1,v^2_{p}) \leq d_G (v_1, v_p) + \frac{M}{2^{k-1}}$
\end{enumerate}
\end{lemma}
\vspace{-.2in}
\begin{proof}
Let $i$, $j$ and $k$ be such that: $v_{p-1}$ and $v_{p+1}$ differ on $i$-th bit and $j$ be the $i$-th bit of $v_{p+1}$ and $k$ be such
that $\frac{M}{2^{k}} < w(v_p, v_{p+1}) \leq \frac{M}{2^{k-1}}$. Hence, edge $(v^1_{p-1}, v^2_p)$ is not present and the edge $(v^1_{p+1}, v^2_p)$ is present
in $G_{i,j,k}$ and so $Last_{G_{i,j,k}} (v^1_1,v^2_{p}) \neq Last_G (v_1, v_p)$, satisfying part 2 of the lemma.

Let us map the path $P$ in Observation~\ref{obs:MinWtCyc}
to the path $P'$ in $G_{i,j,k}$, such that all vertices except $v_p$ are mapped to $V_1$ and $v_p$ is mapped to $V_2$
 (bold path from $v^1_1$ to $v^2_p$ in Figure~\ref{fig:MWCGijk}).
Then, if $P'$ is a shortest path from $v_1^1$ to $v_p^2$ in $G_{i,j,k}$, both parts 1 and 3 of the lemma will hold. So it remains to show that $P'$ is a shortest path. But if not, an
actual shortest path from $v_1^1$ to $v_p^2$ in $G_{i,j,k}$ would create a shorter cycle in $G$ than $C$, and if that cycle were not simple, one could extract from it an even shorter
cycle, contradicting the fact that $C$ is a minimum weight cycle in $G$.
\end{proof}

\begin{lemma}	\label{lemma:MinWtCycAPSP2}
If there exists an $i \in \{1,\ldots,\lceil \log n \rceil\}$, $j \in \{0,1\}$ and
$k \in \{1, 2, \ldots, \lceil \log \rho \rceil \}$ and $y, z \in V$ such that the following conditions hold:

\vspace{-.175in}
\begin{enumerate}
\item $d_{G_{i,j,k}} (y^1,z^2) + d_G (y, z) = wt$ for some $wt$
\vspace{-.075in}
\item $Last_{G_{i,j,k}} (y^1, z^2) \neq Last_G (y, z)$
\vspace{-.075in}
\item $d_{G_{i,j,k}} (y^1, z^2) \leq d_G (y,z) + \frac{M}{2^{k-1}}$
\end{enumerate}
\vspace{-.175in}
Then there exists a simple cycle $C$ containing $z$ of weight at most $wt$ in $G$.	
\end{lemma}

\vspace{-.25in}
\begin{proof}
Let $\pi_{y,z}$ be a shortest path from $y$ to $z$ in $G$
 (see Figure~\ref{fig:MWCCycle})
 and let $\pi_{y^1,z^2}$ be a shortest path from $y^1$ to $z^2$ in $G_{i,j,k}$
 (Figure~\ref{fig:MWCGijk}).
Let $\pi^{'}_{y,z}$ be the path corresponding to $\pi_{y^1,z^2}$ in $G$.

Now we need to show that the path $\pi^{'}_{y,z}$ is simple.
Assume that $\pi^{'}_{y,z}$ is not simple.
It implies that the path $\pi_{y^1,z^2}$
must contain $x^1$ and $x^2$ for some $x \in V$.
Now if $x \neq z$, then we can remove the subpath from $x^1$ to $x^2$ (or from $x^2$ to $x^1$)
 to obtain an even shorter path from $y^1$ to $z^2$.

It implies that the path $\pi_{y^1,z^2}$ contains $z^1$ as an internal vertex.
Let $\pi_{z^1,z^2}$ be the subpath of $\pi_{y^1,z^2}$ from vertex $z^1$ to $z^2$.
If $\pi_{z^1,z^2}$ contains at least 2 internal vertices then 
this would be a simple cycle of weight less than $wt$, and we are done.
Otherwise, the path $\pi_{z^1,z^2}$ contains exactly one internal vertex (say $x^1$).
Hence path $\pi_{z^1,z^2}$ corresponds to the edge $(z,x)$ traversed twice in graph $G$.
But the weight of the edge $(x,z)$ must be greater than $\frac{M}{2^k}$ (as the edge $(x^1,z^2)$ is present in $G_{i,j,k}$).
Hence $w(\pi_{z^1,z^2}) > \frac{M}{2^{k-1}} $ and hence 
$d_{G_{i,j,k}} (y^1, z^2) \geq d_G (y, z) + w(\pi_{z^1,z^2}) > d_G (y, z) + \frac{M}{2^{k-1}}$, resulting in a contradiction as
condition 3 states otherwise. 
(It is for this property that the index $k$ in $G_{i,j,k}$ is used.)
Thus path $\pi_{y^1,z^2}$ does not contain $z^1$ as an internal vertex and hence $\pi^{'}_{y,z}$ is simple. 

If the paths $\pi_{y,z}$ and $\pi^{'}_{y,z}$ do not have any internal vertices in common, then
$\pi_{y,z} \circ \pi^{'}_{y,z}$ corresponds to a simple cycle $C$ in $G$ of weight $wt$ that passes through $y$ and $z$.
Otherwise, we can extract from $\pi_{y,z} \circ \pi^{'}_{y,z}$ a cycle of weight smaller than $wt$. This establishes the lemma.
\end{proof}

\begin{proof}[Proof of Theorem~\ref{thm1}:]
To compute the weight of a minimum weight cycle in $G$ in 
$\tilde O(n^2 + T_{APSP})$, 
we 
use the procedure 
MWC-to-APSP
described in Section~\ref{sec:undir}.
By Lemmas~\ref{lemma:MinWtCycAPSP1}
and \ref{lemma:MinWtCycAPSP2}, the value $wt$ returned by this algorithm is the weight of a minimum
weight cycle in $G$.
\end{proof}

\vspace{-.3in}
\subparagraph*{Sparse Reduction to APSD:} 
We now describe how to avoid using the  
$Last$
matrix  in the reduction.
A  2-approximation algorithm for 
finding a cycle of weight at most $2t$, where $t$ is 
such that the minimum-weight cycle's weight lies in the range $(t, 2t]$, as well as
distances between pairs of vertices within distance at most $t$, was given
by Lingas and Lundell~\cite{LL09}. 
This algorithm can also compute
the last edge on each shortest path it computes,
and its running time is $\tilde{O}(n^2\log (n \rho))$.
For a minimum weight cycle $C = \langle v_1, v_2, \ldots, v_l \rangle$ where the edge $(v_p,v_{p+1})$ is a critical
edge with respect to the start vertex $v_1$, the shortest path length
 from $v_1$ to $v_p$ or to $v_{p+1}$ is at most $t$.
Thus using this algorithm, we can compute the last edge on a shortest path for
 such pair of vertices in $\tilde{O}(n^2\log (n \rho))$ time.

In our reduction to 
APSD,
we first run the 2-approximation algorithm on the input graph $G$ to obtain
the $Last(y,z)$  for certain pairs of vertices.
Then,  in Step 5 we  check if $Last_{G_{i,j,k}} (y^1, z^2) \neq Last_G (y, z)$ {\it only if} $Last_G (y, z)$ has been
computed (otherwise the current path is not a candidate for computing a minimum weight cycle).
It appears from the algorithm that the $Last$ values are also needed in the $G_{i,j,k}$. However,
instead of computing the $Last$ values in each $G_{i,j,k}$,
 we check for the shortest path from $y$ to $z$ only in those $G_{i,j,k}$ graphs
where the $Last_G (y,z)$ has been computed, and the edge is not present in $G_{i,j,k}$.
In other words, if 
$Last(y,z) = q$,
we will only consider the shortest paths from $y^1$ to $z^2$ in those graphs $G_{i,j,k}$ 
where $q$'s $i$-th
bit is not equal to $j$.
Thus our reduction to 
APSD
goes through without needing 
APSP
to output the $Last$ matrix.
 $\square$

\subsection{Reducing ANSC to APSP in Unweighted Undirected Graphs}	\label{sec:ANSCtoAPSP}

For our sparse $\tilde O(n^2)$ reduction from ANSC to 
APSP
 in unweighted undirected graphs, we
use the graphs from the previous section, but we do not use 
the index $k$, since the graph is unweighted. 

Our reduction exploits the fact that in unweighted graphs, every edge in a cycle is a critical edge with respect to some vertex.
Thus we construct $2\lceil \log n \rceil$ graphs $G_{i,j}$, and in order to construct a shortest cycle
through vertex $z$ in $G$, we will set $z= v_p^2$ in the reduction in the previous section. Then, by letting
one of the two edges incident on $z$ in the shortest cycle through $z$ be the critical  edge for the cycle,
the construction from the previous section will allow us to find the length of a minimum length cycle
through $z$, for each $z\in V$, with the following post-processing algorithm.

\begin{algorithm}[H]	
\scriptsize
\caption*{ANSC-to-APSP}
\begin{algorithmic}[1]
\For{each vertex $z \in V$}	
    \State $wt[z] \leftarrow \infty$
\EndFor
\For{$1\leq i\leq \lceil \log n\rceil$, $j \in \{0,1\}$}
	\State{Compute $\mathrm{APSP'}$ on $G_{i,j,}$}
	\For{$y,z \in V$}
		\State {\bf  if} $d_{G_{i,j,}} (y^1, z^2) \leq d_G (y,z) + 1$ {\bf then} check if $Last_{G_{i,j,}} (y^1, z^2) \neq Last_G (y, z)$ 	\label{algansc:check1} 
		\State {\bf if} both checks in Step~\ref{algansc:check1}  hold {\bf then} $wt[z] \leftarrow min(wt[z], d_{G_{i,j,}} (y^1,z^2) + d_G (y, z))$ 
	\EndFor
\EndFor
\State {\bf return} $wt$ array
\end{algorithmic} 
\end{algorithm}
\vspace{-.15in}

Correctness of the above sparse reduction follows from the following two lemmas, which
are similar to Lemmas~\ref{lemma:MinWtCycAPSP1}
and \ref{lemma:MinWtCycAPSP2}.

\begin{lemma}	\label{lemma:ANSCtoAPSP1}
Let $C = \langle z, v_2, v_3, \ldots, v_q \rangle$ be a minimum length cycle passing through vertex $z \in V$.
Let $(v_p,v_{p+1})$ be its critical edge such that $p = \lfloor \frac{q}{2} \rfloor + 1$.
Then there exists an $i \in \{1,\ldots, \lceil \log n \rceil\}$ and $j \in \{0,1\}$ such that the following conditions hold:
\vspace{-.175in}
\begin{enumerate}[(i)]
\item $d_{G_{i,j}} (v^1_p,z^2) + d_G  (v_p,z) = len(C)$
\vspace{-.075in}
\item $Last_{G_{i,j}} (v^1_p,z^2) \neq Last_G (v_p,z)$
\vspace{-.075in}
\item $d_{G_{i,j}} (v^1_p,z^2) \leq d_G (v_p,z) + 1$
\end{enumerate}
\end{lemma}

\begin{lemma}	\label{lemma:ANSCtoAPSP2}
If there exists an $i \in \{1, \ldots, \lceil \log n \rceil \}$ and $j \in \{0,1\}$ and $y, z \in V$ such that the following conditions hold:
\vspace{-.175in}
\begin{enumerate}[(i)] 
\item $d_{G_{i,j}} (y^1,z^2) + d_G (y,z) = q$ for some $q$ where $d_G (y,z) = \lfloor \frac{q}{2} \rfloor$ 
\vspace{-.075in}
\item $Last_{G_{i,j}} (y^1,z^2) \neq Last_G (y,z)$
\vspace{-.075in}
\item $d_{G_{i,j}} (y^1,z^2) \leq d_G (y,z) + 1$
\end{enumerate}
\vspace{-.1in}
Then there exists a simple cycle $C$ passing through $z$ of length at most $q$ in $G$.
\end{lemma}

\vspace{-.175in}
\begin{proof}[Proof of Theorem~\ref{thm2}:]
We now show that the  entries
in the $wt$ array returned by the above algorithm correspond  to the ANSC output
for $G$.
Let $z \in V$ be an arbitrary vertex in $G$ and let $q = wt[z]$.
Let $y'$ be the vertex in Step 5 for which we obtain this value of $q$.
Hence by Lemma~\ref{lemma:ANSCtoAPSP2}, there exists a simple cycle $C$ passing through $z$ of length at most $q$ in $G$. If there were a cycle through $z$ of length $q'<q$ then 
by Lemma~\ref{lemma:ANSCtoAPSP1}, there exists a vertex $y''$ such that conditions in Step~\ref{algansc:check1} hold for $q'$, and the algorithm would have returned a smaller value than
$wt[z]$, which is a contradiction.
This is a sparse $\tilde{O}(n^2)$ reduction since it makes $O(\log n)$ calls to 
 APSP,
and spends $\tilde{O}(n^2)$ additional time.
\end{proof}

It would be interesting to see if we can obtain a reduction from {\it  weighted} ANSC to 
 APSD or APSP.
The
above reduction does not work for the weighted case since it exploits the fact that for any cycle $C$ through
a vertex $z$, an edge in $C$ that is incident on $z$ is a critical edge for some vertex in $C$. 
However, this property
need not hold in the weighted case.
\section{ Fine-grained Reductions for Directed Graphs}\label{sec:otherdir}

In Section~\ref{sec:dir}, we gave an overview of our tilde-sparse $O(m+n\log n)$ reduction from directed $2$-SiSP to the Radius problem.
Here we give full details of this reduction along with the rest of our sparse reductions for directed graphs
 (except
for
 the reductions related to Centrality problems which are described in Section~\ref{sec:centrality}).

\subparagraph{I. \underline{2-SiSP to Radius and $s$-$t$ Replacement Paths to Eccentricities:}}
Here we give the details of the sparse reduction from 2SiSP to Radius described in Section~\ref{sec:dir}
and a related sparse reduction from $s$-$t$ Replacement Paths to Eccentricities.

We start by pointing out some differences between the sparse reduction we give below and an earlier sub-cubic reduction to Radius in~\cite{AGW15}.
In that sub-cubic reduction to Radius, the starting problem is Min-Wt-$\Delta$.
This reduction 
constructs
a 4-partite graph 
that contains paths that correspond to triangles
in the original graph.
However, Min-Wt-$\Delta$ can be solved in $O(m^{3/2})$ time and hence 
to use this reduction to show 
MWC hardness for Radius,
we first need to show that Min-Wt-$\Delta$ is MWC Hard.
 But 
 this
 would achieve a major
breakthrough by giving an $\tilde{O}(m^{3/2})$ time algorithm for 
MWC
(and would refute MWCC).
If we tried to adapt this reduction in~\cite{AGW15}
to the $mn$ class by starting from MWC instead of Min-Wt-$\Delta$, it appears that we would need to have an $n$-partite graph,
since a minimum weight cycle could pass through all $n$ vertices in the graph. This would not be a sparse reduction.
Hence, here we instead present a more complex reduction from 2-SiSP.
As in~\cite{AGW15},
 we also use
 bit-encoding to preserve 
sparsity in the reduced graph.
However the rest of the reduction is different.

\begin{lemma}	\label{lemma:2SiSPRadius}
In weighted directed graphs, 2-SiSP $\lesssimsp_{m+n}$ Radius and $s$-$t$ Replacement Paths $\lesssimsp_{m+n}$ Eccentricities
\end{lemma}

\vspace{-.2in}
\begin{proof}
We are given an input graph $G = (V,E)$, a source vertex $s$ and a sink/target vertex $t$ and we wish to compute the second simple 
shortest path from $s$ to $t$.
Let $P$ ($s=v_0\rightarrow v_1\rightsquigarrow v_{l-1}\rightarrow v_l=t$) be the shortest path from $s$ to $t$ in $G$.

\textit{Constructing the reduced graph $G''$:}
We first create the graph $G'$, which contain $G$ and $l$ additional vertices $z_0$,$z_1$,$\ldots$,$z_{l-1}$.
We remove the edges lying on $P$ from $G'$.
For each $0\leq i\leq l-1$, we add an edge from $z_i$ to $v_i$ of weight $d_G (s,v_i)$ and an edge from $v_{i+1}$ to $z_i$ of weight 
$d_G (v_{i+1},t)$.
Also for each $1\leq i\leq l-1$, we add a zero weight edge from $z_i$ to $z_{i-1}$.

Now form $G''$ from $G'$.
For each $0\leq j\leq l-1$, we replace vertex $z_j$ by vertices $z_{j_i}$ and $z_{j_o}$ and we place a directed edge of weight $0$ from
$z_{j_i}$ to $z_{j_o}$, and we also replace each incoming edge to (outgoing edge from) $z_j$ with an incoming edge to $z_{j_i}$ 
(outgoing edge from $z_{j_o}$) in $G'$.

Let $M$ be the largest edge weight in $G$ and let $M' = 9nM$.
For each $0\leq j\leq l-1$, we add additional vertices $y_{j_i}$ and $y_{j_o}$ and we place a directed edge of weight $0$ 
from $y_{j_o}$ to $z_{j_o}$ and an edge of weight $\frac{11}{9}M'$ from $z_{j_i}$ to $y_{j_i}$.

We add 2 additional vertices $A$ and $B$, and we place a directed edge from $A$ to $B$ of weight $0$.
We also add $l$ incoming edges to $A$ (outgoing edges from $B$) from (to) each of the $y_{j_o}'s$ of weight $0$ ($M'$).

We also add edges of weight $\frac{2M'}{3}$ from $y_{j_o}$ to $y_{k_i}$ (for each $k\neq j$).
But due to the addition of $O(n^2)$ edges, graph $G'$ becomes dense.
To solve this problem, we add a gadget in our construction that ensures that $\forall 0\leq j\leq l-1$, we have at least one path of length 2
and weight equal to $\frac{2M'}{3}$ from $y_{j_o}$ to $y_{k_i}$ (for each $k\neq j$) 
(similar to \cite{AGW15}).
In this gadget, 
we add $2\lceil \log n\rceil$ vertices of the form $C_{r,s}$ for $1\leq r\leq \lceil \log n\rceil$ and $s \in \{0,1\}$.
Now for each $0\leq j\leq l-1$, $1\leq r\leq \lceil \log n\rceil$ and $s\in \{0,1\}$, we add an edge of weight $\frac{M'}{3}$ from 
$y_{j_o}$ to $C_{r,s}$ if $j's$ $r$-th bit is equal to $s$.
We also add an edge of weight $\frac{M'}{3}$ from $C_{r,s}$ to $y_{j_i}$ if $j$'s $r$-th bit is not equal to $s$.
So overall we add $2n\log n$ edges that are incident to $C_{r,s}$ vertices; 
for each $y_{j_o}$ we add $\log n$ outgoing edges to $C_{r,s}$ vertices and 
for each $y_{j_i}$ we add $\log n$ incoming edges from $C_{r,s}$ vertices.

We can observe that for $0\leq j\leq l-1$, there is at least one path of weight $\frac{2M'}{3}$ from $y_{j_o}$ to $y_{k_i}$ (for each 
$k\neq j$) and the gadget does not add any new paths from $y_{j_o}$ to $y_{j_i}$. 
The reason is that for every distinct $j,k$, there is at least one bit (say $r$) where $j$ and $k$ differ and let $s$ be the $r$-th bit of $j$.
Then there must be an edge from $y_{j_o}$ to $C_{r,s}$ and an edge from $C_{r,s}$ to $y_{k_i}$, resulting in a path of weight 
$\frac{2M'}{3}$ from $y_{j_o}$ to $y_{k_i}$.
And by the same argument we can also observe that this gadget does not add any new paths from $y_{j_o}$ to $y_{k_i}$.

We call this graph as $G''$.
Figure~\ref{fig:2SiSPRadius} depicts the full construction of $G''$ for $l = 3$.
We now establish the following three properties.

{\it (i) For each $0\leq j\leq l-1$, the longest shortest path in $G''$ from $y_{j_o}$ is to the vertex $y_{j_i}$.}
It is easy to see that
the shortest path from $y_{j_o}$ to any of the vertices in $G$ or any of the $z$'s has weight at most 
$nM$.
And the shortest paths from $y_{j_o}$ to the vertices $A$ and $B$ have weight $0$.
For $k\neq j$, the shortest path from $y_{j_o}$ to $y_{k_o}$ and $y_{k_i}$ has weight $M'$ and $\frac{2}{3}M'$
respectively.
Whereas the shortest path from $y_{j_o}$ to $y_{j_i}$ has weight at least $10nM$ as it includes the last edge 
$(z_{j_i}, y_{j_i})$ of weight $\frac{11}{9}M' = 11nM$.
It is easy to observe that the shortest path from $y_{j_o}$ to $y_{j_i}$ corresponds to the shortest path from $z_{j_o}$ to $z_{j_i}$.

{\it (ii) The shortest path from $z_{j_o}$ to $z_{j_i}$ corresponds to the  replacement path for the edge $(v_j,v_{j+1})$ lying on $P$.}
Suppose not and let $P_j$ ($s\rightsquigarrow v_h\rightsquigarrow v_k\rightsquigarrow t$) (where $v_h$ is the vertex where $P_j$
separates from $P$
and $v_k$ is the vertex where it joins $P$) be the replacement path from $s$ to $t$ for the edge $(v_j,v_{j+1})$.
But then the path $\pi_j$ ($z_{j_o}\rightarrow z_{{j-1}_i}\rightsquigarrow z_{h_o}\rightarrow v_{h} \circ P_j(v_h,v_k) \circ v_k
\rightarrow z_{k_i}\rightarrow z_{k_o}\rightsquigarrow z_{j_i}$) (where $P_j(v_h,v_k)$ is the subpath of $P_j$ from $v_j$ to $v_k$)
from $z_{j_o}$ to $z_{j_i}$ has weight equal to $wt(P_j)$, resulting in a contradiction as the shortest path from $z_{j_o}$ to $z_{j_i}$ has weight greater than that of $P_j$.

{\it (iii) One of the vertices among $y_{j_o}$'s is a center of $G''$.}
It is easy to
see that none of the vertices in $G$ could be a center of the graph $G''$ as there is no path from any 
$v \in V$ to any of the $y_{j_o}$'s  in $G''$.
Using a similar argument, we can observe that none of the $z$'s, or the vertices $y_{j_i}$'s
could be a potential candidate for the center of $G''$.
For vertices $A$ and $B$, the shortest path to any of the $y_{j_i}$'s has weight exactly $\frac{5}{3}M' = 15nM$, which is strictly greater 
than the weight of the largest shortest path from any of the $y_{j_o}$'s.
Thus one of the vertices among $y_{j_o}$'s is a center of $G''$.

Thus by computing the radius in $G''$, from 
$(i)$, $(ii)$, and $(iii)$,
we can compute the weight of the shortest replacement
path from $s$ to $t$, which by definition of 2-SiSP, is the second simple shortest path from $s$ to $t$.
This completes the proof
of 2-SiSP $\lesssimsp_{m+n}$ Radius.

Now, if instead of computing Radius in $G''$ we compute
 the Eccentricities of all vertices in $G''$, then from 
 $(i)$ and $(ii)$
 we can compute the weight of the 
replacement path for every edge $(v_j,v_{j+1})$ lying on $P$, thus solving the replacement paths problem.
So $s$-$t$ Replacement Paths $\lesssimsp_{m+n}$ Eccentricities.

Constructing $G''$ takes $O(m+n\log n)$ time since we add $O(n)$ additional vertices and $O(m+n\log n)$ additional 
edges, and given the output of Radius (Eccentricities), we can compute 2-SiSP 
($s$-$t$ Replacement Paths) in $O(1)$ ($O(n)$) time and hence 
the cost of both reductions is $O(m+n\log n)$.
\end{proof}

\subparagraph{II. \underline{ANSC and Replacement Paths:}}

We first describe a sparse reduction from directed MWC to 2-SiSP,
which we will use 
for reducing ANSC to the $s$-$t$ replacement paths problem.
This reduction is adapted from a sub-cubic 
non-sparse reduction from Min-Wt-$\Delta$ to 2-SiSP in~\cite{WW10}.
The reduction in~\cite{WW10}
reduces Min-Wt-$\Delta$ to 2-SiSP
  by creating a 
tripartite graph. 
Since starting from Min-Wt-$\Delta$ is not appropriate for our results (as discussed in our sparse reduction to directed Radius), we start instead from MWC, and
instead of the tripartite graph 
used in~\cite{WW10}
we use the original graph $G$ 
with every vertex $v$ replaced 
with 2 copies,
$v_i$ and $v_o$.    

In this reduction,
as in~\cite{WW10},
 we first create a path of length $n$ with vertices labeled from $p_0$ to $p_n$,
 which will be the initial shortest path.
We then map every edge $(p_i,p_{i+1})$ to the vertex $i$ in the original graph $G$ such that the replacement path from
$p_0$ to $p_n$ for the edge $(p_i,p_{i+1})$ corresponds to the shortest cycle passing through $i$ in $G$.
Thus computing 2-SiSP (i.e., the shortest replacement path) from $p_0$ to $p_n$ in the constructed graph corresponds to the minimum weight cycle in the original graph.
The details of the reduction are in the proof of the lemma below.

\begin{lemma}	\label{lemma:MinWtCyc2SiSP}
In weighted directed graphs, MWC $\leqsp_{m+n}$ 2-SiSP
\end{lemma}

\begin{proof}
To compute MWC in $G$, we first create the graph $G'$, where we replace every vertex $z$ by vertices $z_i$ and $z_o$, and we
place a directed edge of weight 0 from $z_i$ to $z_o$, and we replace each incoming edge to (outgoing edge from) $z$ with an incoming
edge to $z_i$ (outgoing edge from $z_o$).
We also add a path $P$ ($p_0 \rightarrow p_1 \rightsquigarrow p_{n-1} \rightarrow p_n$)  of length $n$ and weight $0$.

Let $Q = n \cdot$ \Wmax,
where \Wmax is the maximum weight of any edge in $G$.  
For each $1 \leq j \leq n$,  we add an
edge of weight $(n-j+1)Q$ from $p_{j-1}$ to $j_o$ and an edge of weight $jQ$ from $j_i$ to $p_j$ in $G'$  to form $G''$. 
Figure \ref{fig:hardA} depicts the full construction of $G''$ for $n=3$. This is an $(m+n)$ reduction, and it can be seen that the second simple shortest path from $p_0$ to $p_n$ in $G''$ corresponds to a minimum weight cycle in $G$. 

\begin{figure}
\scriptsize
    \centering
        \makebox[.3\textwidth]{
            \begin{tikzpicture}[scale=0.5, every node/.style={circle, draw, inner sep=0pt, minimum width=5pt}]
            \node (1o)[label=above:$1_o$] at (0,0)  {};
            \node (1i)[label=above:$1_i$] at (1,0) {};
            \node (2o)[label=above:$2_o$]	at (0,-1)		{};
            \node (2i)[label=above:$2_i$] at (1,-1)	{};
            \node (3o)[label=above:$3_o$] at (0,-2)	{};
            \node (3i)[label=above:$3_i$] at (1,-2) {};
            \node (p0)[label=above:$p_0$] at (-3,-4) {};
            \node (p1)[label=above:$p_1$] at (-1,-4) {};
            \node (p2)[label=above:$p_2$] at (1,-4) {};
            \node (p3)[label=above:$p_3$] at (3,-4) {};
            \node[ellipse, color=blue, fit=(1o)(1i)(2o)(2i)(3o)(3i), inner sep=3mm](G') [label=left:$G'$] {};
            \begin{scope}[->,every node/.style={ right}]
            \draw[->] (p0) -- (p1) node[midway, label=below:$0$] {};
            \draw[->] (p1) -- (p2) node[midway, label=below:$0$] {};
            \draw[->] (p2) -- (p3) node[midway, label=below:$0$] {};
            \draw[->,color=red] (p0) -- (1o) node[midway, label=left:$3W$] {};
            \draw[->,color=red] (p1) -- (2o) node[midway, label=left:$2W$] {};
            \draw[->,color=red] (p2) -- (3o) node[midway, label=left:$W$] {};
            \draw[->,color=black] (1i) to  [bend left=120, distance=15mm]  (p1) [dashed] node[label=above:$W$] at (2,-0.5) {};
            \draw[->,color=black] (2i) to  [bend left=90, distance=15mm]  (p2) [dashed] node[label=right:$2W$] at (2,-1.5) {};
            \draw[->,color=black] (3i) to  [bend left=60, distance=10mm]  (p3) [dashed] node[label=right:$3W$] at (3,-3) {};
            \end{scope}
            \end{tikzpicture}}
    \caption{$G''$ for $n = 3$ in the reduction: directed MWC $\leqsp_{m+n}$ 2-SiSP}
    \label{fig:hardA}
\end{figure}
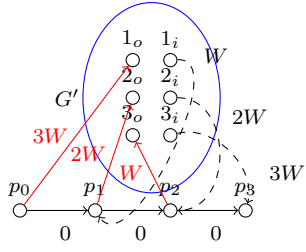
\vspace{-.2in}
\end{proof}

We now establish the equivalence between ANSC and the $s$-$t$ replacement paths problem under $(m+n)$-reductions by first showing an 
$(m+n)$-sparse reduction from $s$-$t$ replacement paths problem to ANSC. 
We then describe a sparse reduction from ANSC to the $s$-$t$ replacement paths problem, which is similar to the reduction from 
MWC to 2-SiSP as described in Lemma~\ref{lemma:MinWtCyc2SiSP}.

\begin{lemma}	\label{lemma:ReplacementANSC}
In weighted directed graphs, $s$-$t$ replacement paths  $\equivsp_{m+n}$ ANSC
\end{lemma}
\vspace{-.2in}
\begin{proof}
We are given an input graph $G = (V,E)$, a source vertex $s$ and a sink vertex $t$ and we wish to compute the replacement paths for all
the edges lying on the shortest path from $s$ to $t$.
Let $P$($s=v_0\rightarrow v_1\rightsquigarrow v_{l-1}\rightarrow v_l=t$) be the shortest path from $s$ to $t$ in $G$.

$(i)$ \textit{Constructing $G'$:}
We first create the graph $G'$, 
as described in the proof of Lemma~\ref{lemma:2SiSPRadius}.
Figure~\ref{fig:hardB} depicts the full construction of $G'$ for $l=3$.

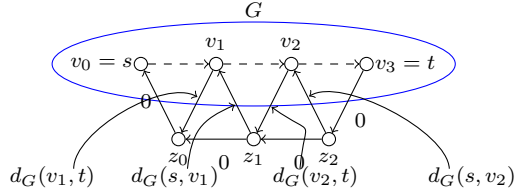
\begin{figure}
\scriptsize
    \centering
        \makebox[.3\textwidth]{
            \begin{tikzpicture}[scale=0.5, every node/.style={circle, draw, inner sep=0pt, minimum width=5pt}]
            \node (v0)[label=left:${v_0=s}$] at (0,0)  {};
            \node (v1)[label=above:$v_1$] at (2,0) {};
            \node (v2)[label=above:$v_2$]	at (4,0)	{};
            \node (v3)[label=right:${v_3=t}$] at (6,0)	{};
            \node[ellipse, color=blue, fit=(v0)(v1)(v2)(v3), inner sep=3mm](G) [label=above:$G$] {};
            \node (z0)[label=below:$z_0$] at (1,-2)  {};
            \node (z1)[label=below:$z_1$] at (3,-2)  {};
            \node (z2)[label=below:$z_2$] at (5,-2)  {};
            \draw[->] (v0) -- (v1) [dashed] {};
            \draw[->] (v1) -- (v2) [dashed] {};
            \draw[->] (v2) -- (v3) [dashed] {};
            \begin{scope}[->,every node/.style={ right}]
            \draw[->] (z0) -- (v0) node[midway, label=left:$0$] {};          
            \draw[->] (v1) -- (z0) {};
            \node (labelv1z0)[label=left:${d_G(v_1,t)}$] at (-1,-3) {};
            \node (v1z0start) at (-2,-3) {};
            \node (v1z0end) at (1.5,-1) {};
            \draw[->] (v1z0start) to  [bend left=60, distance=10mm]  (v1z0end) {};
            \draw[->] (z1) -- (z0) node[midway, label=below:$0$] {}; 
            \draw[->] (z1) -- (v1) {}; 
            \node (labelz1v1)[label=right:${d_G(s,v_1)}$] at (-0.9,-3) {};
            \node (z1v1start) at (1,-3) {};
            \node (z1v1end) at (2.5,-1) {};
            \draw[->] (z1v1start) to  [bend left=30, distance=10mm]  (z1v1end) {};
            \draw[->] (v2) -- (z1) {};
            \node (labelv2z1)[label=left:${d_G(v_2,t)}$] at (6,-3) {};
            \node (v2z1start) at (4,-3) {};
            \node (v2z1end) at (3.05,-1) {};
            \draw[->] (v2z1start) to  [bend right=30, distance=15mm]  (v2z1end) {};
            \draw[->] (z2) -- (z1) node[midway, label=below:$0$] {};
            \draw[->] (z2) -- (v2) {}; 
            \node (labelz2v2)[label=right:${d_G(s,v_2)}$] at (7,-3) {};
            \node (z2v2start) at (8,-3) {};
            \node (z2v2end) at (4.05,-1) {};
            \draw[->] (z2v2start) to  [bend right=60, distance=10mm]  (z2v2end) {};
            \draw[->] (v3) -- (z2) {}; 
            \node (v3z2label)[label=right:$0$] at (5.05,-1.5) {};
            \end{scope}
            \end{tikzpicture}}
    \caption{$G'$ for $l = 3$ in the reduction: directed $s$-$t$ Replacement Paths  $\leqsp_{m+n}$ ANSC}
    \label{fig:hardB}
\end{figure}

$(ii)$ \textit{We now show that for each $0 \leq i \leq l-1$, the replacement path from $s$ to $t$ for the edge $(v_i, v_{i+1})$ lying on 
$P$ has weight equal to the shortest cycle passing through $z_i$.}
 If not, assume that for some $i$ ($0\leq i\leq l-1$), the weight of the replacement path from $s$ to $t$ for the edge 
$(v_i,v_{i+1})$ is not equal to the weight of the shortest cycle passing through $z_i$.

Let $P_i$ ($s\rightsquigarrow v_j\rightsquigarrow v_k\rightsquigarrow t$) (where $v_j$ is the vertex where $P_i$ separates from $P$
and $v_k$ is the vertex where it joins $P$) be the replacement path from $s$ to $t$ for the edge $(v_i,v_{i+1})$ and let $C_i$ 
($z_i\rightsquigarrow z_p\rightarrow v_p\rightsquigarrow v_q\rightarrow z_q\rightsquigarrow z_i$) be the shortest cycle passing through
$z_i$ in $G'$.

If $wt(P_i) < wt(C_i)$, then the cycle $C'_i$ ($z_i\rightarrow z_{i-1}\rightsquigarrow z_j\rightarrow v_{j} \circ P_i(v_j,v_k) \circ v_k
\rightarrow z_k\rightarrow z_{k-1}\rightsquigarrow z_i$) (where $P_i(v_j,v_k)$ is the subpath of $P_i$ from $v_j$ to $v_k$) passing 
through $z_i$ has weight equal to $wt(P_i) < wt(C_i)$, resulting in a contradiction as $C_i$ is the shortest cycle passing through $z_i$ in
$G'$.

Now if $wt(C_i) < wt(P_i)$, then the path $P'_i$ ($s\rightsquigarrow v_p \circ C_i(v_p,v_q) \circ v_q\rightsquigarrow v_l$) where 
$C_i(v_p,v_q)$ is the subpath of $C_i$ from $v_p$ to $v_q$, is also a path from $s$ to $t$ avoiding the edge $(v_i,v_{i+1})$, and has 
weight equal to $wt(C_i) < wt(P_i)$, resulting in a contradiction as $P_i$ is the shortest replacement path from $s$ to $t$ for the edge 
$(v_i,v_{i+1})$.

We then compute ANSC in $G'$. 
And by (ii), the shortest cycles  for each of the vertices $z_0, z_1, \ldots, z_{l-1}$ gives us the replacement paths from $s$ to
$t$.
This leads to an $(m+n)$ sparse reduction from $s$-$t$ replacement paths problem to ANSC.

Now for the other direction, we are given an input graph $G = (V,E)$ and we wish to compute the ANSC in $G$.
We first create the graph $G^{''}$, as described in Lemma~\ref{lemma:MinWtCyc2SiSP}.
We can see that the shortest path from $p_0$ to $p_n$ avoiding edge $(p_{j-1},p_j)$ corresponds to a shortest cycle passing through $j$
in $G$.
This gives us an $(m+n)$-sparse reduction from ANSC to $s$-$t$ replacement paths problem.
\end{proof}

\section{Betweenness Centrality: Reductions}	\label{sec:centrality}

In this section, we consider sparse reductions for Betweenness Centrality and related problems.
In its full generality, the Betweenness Centrality of a vertex $v$ is the sum, across all pairs of vertices $s,t$,
 of the fraction of shortest paths from $s$ to $t$ 
 that contain $v$ as an internal vertex.
 This problem has a $\tilde{O}(mn)$ time algorithm
 due to Brandes~\cite{Brandes01}. 
   Since there
 can be an exponential (in $n$) number of shortest paths from one vertex to another,
 this general problem can deal with very large numbers. In~\cite{AGW15}, a simplified variant
 was considered, where
 it is assumed that there is a unique shortest path for each pair of vertices, and 
  the Betweenness Centrality of vertex $v$, $BC(v)$, is defined as
the number of vertex pairs $s, t$  such that  $v$ is an internal vertex
 on the unique shortest path from $s$ to
$t$. We will also restrict our attention to this variant here.

A  number of sparse reductions relating to the following problems were given in~\cite{AGW15}.

\vspace{-.1in}
\begin{itemize}
\item {\it Betweenness Centrality (BC)} of a vertex $v$, $BC(v)$.

\item {\it Positive Betweenness Centrality (Pos BC) of $v$:} determine whether $BC(v) > 0$.

\item {\it All Nodes Betweenness Centrality (ANBC):} compute, for each $v$, the value of $BC(v)$.

\item {\it Positive All Nodes Betweenness Centrality (Pos ANBC):} determine, for each $v$, whether $BC(v) > 0$.

\item {\it Reach Centrality (RC)} of $v$:
compute $\max_{s,t \in V: \\ d_G (s,v) + d_G (v,t) = d_G (s,t)} {\min (d_G (s,v), d_G (v,t))}$.
\end{itemize}

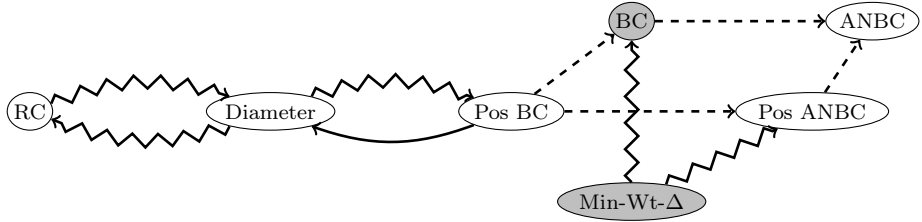
\begin{figure}
\scriptsize
	\centering
            	\begin{tikzpicture}[scale=0.4, every node/.style={circle, draw, inner sep=0pt, minimum width=5pt, align=center}]
            	\node[ellipse, minimum size=0.5cm] (RC) at (0,0) {RC};
            	\node[ellipse, minimum size=0.5cm] (Diameter) at (8,0) {Diameter};
            	\node[ellipse, minimum size=0.5cm] (BCPos) at (16,0) { Pos BC};
            	\node[ellipse, minimum size=0.5cm] (ANBCPos) at (26,0) {Pos ANBC};
            	\node[ellipse, minimum size=0.5cm,fill=lightgray] (BC) at (20,3) {BC};
            	\node[ellipse, minimum size=0.5cm] (ANBC) at (28,3) {ANBC};            	
            	\node[ellipse, minimum size=0.5cm,fill=lightgray] (MWT) at (20,-3) {Min-Wt-$\Delta$};
            	\draw[->, line width=1pt, dashed] (BCPos) -- (ANBCPos) {};
            	\draw[->, line width=1pt, dashed] (BCPos) -- (BC) {};            	            	
            	\draw[->, line width=1pt, dashed] (BC) -- (ANBC) {};            	            	
            	\draw[->, line width=1pt, dashed] (ANBCPos) -- (ANBC) {};        
            	\draw[->, line width=1pt, decorate, decoration={zigzag}] (MWT) -- (BC) {};            	            	    	            	
             	\draw[->, line width=1pt, decorate, decoration={zigzag}] (MWT) -- (ANBCPos) {};            	            	    	            	
             	\draw[->,  line width=1pt, decorate, decoration={zigzag}] (RC) to [bend left=20] (Diameter) {};
            	\draw[->,  line width=1pt, decorate, decoration={zigzag}] (Diameter) to [bend left=20] (RC) {};         
            	\draw[->,  line width=1pt, decorate, decoration={zigzag}] (Diameter) to [bend left=20] (BCPos) {};            
            	\draw[->,  line width=1pt] (BCPos) to [bend left=20] (Diameter) {};          	
            	\end{tikzpicture}
        \caption{Known sparse reductions for centrality problems, all from~\cite{AGW15}. The regular edges represent sparse $O(m+n)$ reductions, the squiggly edges represent tilde-sparse $O(m+n)$ reductions, and the dashed edges represent reductions that are trivial. BC and Min-Wt-$\Delta$ (shaded with gray) are known to be sub-cubic equivalent to APSP~\cite{AGW15,WW10}.}
    \label{fig:priorReductionsBC}
\end{figure}

Figure~\ref{fig:priorReductionsBC} gives an overview of the previous fine-grained results given in~\cite{AGW15} for Centrality problems. 
In this figure, BC is the only centrality problem that is known to be sub-cubic equivalent to APSP, and hence
is shaded in the figure (along with Min-Wt-$\Delta$).  None of these sparse reductions 
in~\cite{AGW15}
imply MWCC hardness for any of the centrality problems since Diameter is not MWCC-hard (or
even sub-cubic equivalent to APSP), and Min-Wt-$\Delta$ has an $\tilde{O}(m^{3/2})$ time algorithm,
and so will falsify MWCC if it is MWCC-hard.
On other hand, Diameter is known to be both SETH-hard~\cite{RW13} and $k$-DSH Hard (Section~\ref{sec:seth}) and
hence all these problems in Figure~\ref{fig:priorReductionsBC} (except Min-Wt-$\Delta$) are also SETH and $k$-DSH hard.

In this section, we give a sparse reduction from 2-SiSP to BC, establishing MWCC-hardness for BC.
We also give a tilde-sparse reduction from ANSC to Pos ANBC, and thus we have MWCC-hardness for both Pos ANBC and for ANBC,
though neither problem is known to be in the sub-cubic equivalence class.  (Both
 have $\tilde{O}(mn)$ time algorithms, and have APSP-hardness under sub-cubic reductions.)

Figure~\ref{fig:reductionsBC} gives 
an updated
partial order 
of our sparse reductions
for weighted directed graphs;
this figure augments Figure~\ref{fig:reductions} by including the sparse reductions for BC problems
given in this section.

 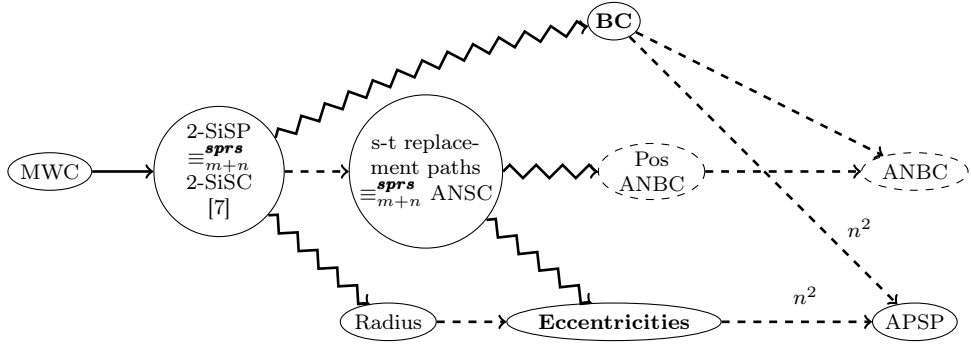
\begin{figure}
\scriptsize
	\centering
\Xomit{
    	\begin{subfigure}{.6\textwidth}
            	\begin{tikzpicture}[scale=0.4, every node/.style={circle, draw, inner sep=0pt, minimum width=5pt, align=center}]
            	\node[ellipse, minimum size=0.5cm] (MWC) at (0,0) {MWC};
            	\node[ellipse, minimum size=0.5cm] (ANSC) at (22,0) {ANSC};
            	\node[ellipse, minimum size=0.5cm] (Radius) at (0,-3) {Radius};
            	\node[ellipse, minimum size=0.5cm] (Eccentricities) at (11,-3) {Eccentricities};
            	\node[ellipse, minimum size=0.5cm] (APSP) at (22,-3) {APSP};
            	\node[ellipse, minimum size=0.8cm] (SETH) at (0,-6) {SETH};
            	\draw[->, line width=1pt, dashed] (MWC) -- (ANSC) {};
            	\begin{scope}[->,every node/.style={ right}]
            	\draw[->,  line width=1pt,decorate, decoration={zigzag}] (MWC) -- (APSP) {};
            	\node (l) at (15,-1.5) {$\tilde{O}(n^2)$};
            	\draw[->, line width=1pt, dashed] (Eccentricities) -- (APSP) node[midway, label=below:$n^2$] {};
            	\end{scope}
            	\draw[->, line width=1pt, dashed] (Radius) -- (Eccentricities) {};
            	\draw[->] (SETH) -- (Eccentricities) {};
            	\end{tikzpicture}
            	\caption{Sparse Reductions for Weighted Undirected Graphs. }
            	\label{fig:reductionsBCUndirected}
		\end{subfigure}    
		~
		\begin{subfigure}{.7\textwidth}
		}
				\begin{tikzpicture}[scale=0.5, every node/.style={circle, draw, inner sep=0pt, minimum width=5pt, align=center}]
            	\node[ellipse, minimum size=0.5cm] (MWC) at (-1,0) {MWC};
            	\node[circle] (2SiSP2SiSC) [text width = 1.2cm] at (3.5,0) {2-SiSP $\equivsp_{m+n}$ 2-SiSC ~\cite{AR16}};
            	\node[circle] (ReplacementANSC) [text width =1.8cm] at (9,0) {s-t replacement paths $\equivsp_{m+n}$ ANSC};
            	\node[ellipse, minimum size=0.5cm,dashed] (ANPBC) [text width = 1cm] at (15,0) {Pos ANBC};
            	\node[ellipse, minimum size=0.5cm,dashed] (ANBC) [text width =1cm] at (22,0) {ANBC};
            	\node[ellipse, minimum size=0.5cm] (Radius) at (8,-4) {Radius};
               \node[ellipse, minimum size=0.5cm] (Eccentricities) at (14,-4) {{\bf Eccentricities}};      
            	\node[ellipse, minimum size=0.5cm] (APSP) at (22,-4) {APSP};
            	\node[ellipse, minimum size=0.5cm] (BC) at (14,4) {{\bf BC}};
            	\draw[->, line width=1pt] (MWC) -- (2SiSP2SiSC) {};
            	\draw[->, line width=1pt, dashed] (2SiSP2SiSC) -- (ReplacementANSC) {};
            	\draw[->, line width=1pt, decorate, decoration={zigzag}] (ReplacementANSC) -- (ANPBC)  {};
            	\draw[->, line width=1pt, dashed] (ANPBC) -- (ANBC) {};
            	\draw[->, line width=1pt, decorate, decoration={zigzag}] (2SiSP2SiSC) -- (Radius) {};
            	\draw[->, line width=1pt, dashed] (Radius) -- (Eccentricities) {};
            	\draw[->, line width=1pt, decorate, decoration={zigzag}] (ReplacementANSC) -- (Eccentricities) {};
            	\begin{scope}[->,every node/.style={ right}]
            	\draw[->, line width=1pt, dashed] (Eccentricities) -- (APSP) node[midway, label=above:$n^2$] {};
            	\node(l) at (20,-1.5) {$n^2$} ;
            	\draw[->, line width=1pt, dashed] (BC) -- (APSP) {};
            	\end{scope}
            	\draw[->, line width=1pt, decorate, decoration={zigzag}] (2SiSP2SiSC) to [bend left=10] (BC)  {};
            	\draw[->, line width=1pt, dashed] (BC) --  (ANBC) {};
            	\end{tikzpicture}
       \caption{Sparse reductions for weighted directed graphs. The regular edges represent sparse $O(m+n)$ reductions, the squiggly edges represent tilde-sparse $O(m+n)$ reductions, and the dashed edges represent reductions that are trivial. All problems except APSP are MWCC-hard. Eccentricities, BC, ANBC and Pos ANBC are also SETH/ $k$-DSH hard. ANBC and Pos ANBC (problems inside the dashed circles) are both not known to be subcubic equivalent to APSP.}
    \label{fig:reductionsBC}
\end{figure}

\subparagraph{I. \underline{2-SiSP to BC}:}
Our sparse reduction from 2-SiSP to BC is similar to the reduction from 2-SiSP to Radius described in Section~\ref{sec:dir}.
In our reduction, we first map every edge $(v_j,v_{j+1})$ to new vertices $y_{j_o}$ and $y_{j_i}$ such that the shortest path from 
$y_{j_o}$ to $y_{j_i}$  corresponds to the replacement path from $s$ to $t$ for the edge $(v_j,v_{j+1})$.
We then add an additional vertex $A$ and connect it to vertices $y_{j_o}$'s and $y_{j_i}$'s.
We also ensure that the only shortest paths passing through $A$ are from $y_{j_o}$ to $y_{j_i}$.
We then do binary search on the edge weights for the edges going from $A$ to $y_{j_i}$'s with oracle calls to the Betweenness Centrality problem,
to compute the  weight of the shortest replacement path from $s$ to $t$, which by definition of 2-SiSP, is the second simple shortest path
from $s$ to $t$.

\begin{lemma}	\label{lemma:2SiSPBetweennessCentrality}
In weighted directed graphs, $2$-SiSP $\lesssimsp_{m+n}$ Betweenness Centrality
\end{lemma}
\vspace{-.25in}
\begin{proof}
We are given an input graph $G = (V,E)$, a source vertex $s$ and a sink/target vertex $t$ and we wish to compute the second simple 
shortest path from $s$ to $t$.
Let $P$ ($s=v_0\rightarrow v_1\rightsquigarrow v_{l-1}\rightarrow v_l=t$) be the shortest path from $s$ to $t$ in $G$.

\textit{(i) Constructing $G''$:}
We first construct the graph $G''$, as described in the proof of Lemma~\ref{lemma:2SiSPRadius}, without the vertices $A$ and $B$.
For each $0\leq j\leq l-1$, we change the weight of the edge from $z_{j_i}$ to $y_{j_i}$ to $M'$ (where $M' = 9nM$ and $M$ is the 
largest edge weight in $G$).

We add an additional vertex $A$ and for each $0\leq j\leq l-1$, we add an incoming (outgoing) edge from (to) $y_{j_o}$ ($y_{j_i}$). 
We assign the weight of the edges from $y_{j_o}$'s to $A$ as $0$ and from $A$ to $y_{j_i}$'s as $M' + q$ (for some $q$ in the range 
$0$ to $nM$).

Figure~\ref{fig:2SiSPBetweennessCentrality} depicts the full construction of $G''$ for $l = 3$.

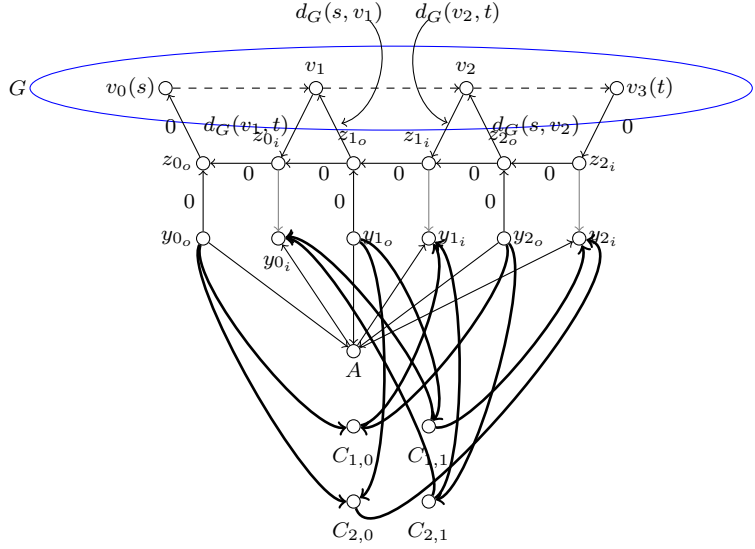
\begin{figure}
\scriptsize
    \centering
        \makebox[.3\textwidth]{
            \begin{tikzpicture}[scale=0.5, every node/.style={circle, draw, inner sep=0pt, minimum width=5pt}]
            \node (v0)[label=left:$v_0(s)$] at (0,0)  {};
            \node (v1)[label=above:$v_1$] at (4,0) {};
            \node (v2)[label=above:$v_2$]	at (8,0)	{};
            \node (v3)[label=right:$v_3(t)$] at (12,0)	{};
            \node[ellipse, color=blue, fit=(v0)(v1)(v2)(v3), inner sep=3mm](G) [label=left:$G$] {};
            \node (z0o)[label=left:$z_{0_o}$] at (1,-2)  {};
            \node (z0i) at (3,-2)  {};
            \node (z1o)[label=above:$z_{1_o}$] at (5,-2)  {};
            \node (z1i) at (7,-2)  {};
            \node (z2o)[label=above:$z_{2_o}$] at (9,-2)  {};
            \node (z2i)[label=right:$z_{2_i}$] at (11,-2)  {};
            \draw[->] (v0) -- (v1) [dashed] {};
            \draw[->] (v1) -- (v2) [dashed] {};
            \draw[->] (v2) -- (v3) [dashed] {};
            \begin{scope}[->,every node/.style={ right}]
            \node (labelz0i)[label=above:$z_{0_i}$] at (2.5,-2) {};
            \node (labelz1i)[label=above:$z_{1_i}$] at (6.5,-2) {};
            \draw[->] (z0o) -- (v0) node[midway, label=left:$0$] {};        
            \draw[->] (z0i) -- (z0o) {};
            \node (labelz0iz0o)[label=below:$0$] at (2,-1.65) {};
            \draw[->] (v1) -- (z0i) node[midway, label=left:${d_G(v_1,t)}$] {};
            \draw[->] (z1o) -- (z0i) {}; 
            \node (labelz1oz0i)[label=below:$0$] at (4,-1.65) {};
            \draw[->] (z1i) -- (z1o) {};
            \node (labelz1iz1o)[label=below:$0$] at (6,-1.65) {};
            \draw[->] (z1o) -- (v1) {};
            \node (labelz1ov1)[label=left:${d_G(s,v_1)}$] at (6,2) {};
            \node (z1ov1start) at (5,2) {};
            \node (z1ov1end) at (4.25,-1) {};
            \draw[->] (z1ov1start) to  [bend left=60, distance=10mm]  (z1ov1end) {};
            \draw[->] (v2) -- (z1i) {};
            \node (labelv2z1i)[label=right:${d_G(v_2,t)}$] at (6,2) {};
            \node (v2z1istart) at (7,2) {};
            \node (v2z1iend) at (7.5,-1) {};
            \draw[->] (v2z1istart) to  [bend right=60, distance=10mm]  (v2z1iend) {};
            \draw[->] (z2o) -- (z1i) {};
            \node (labelz2oz1i)[label=below:$0$] at (8,-1.65) {};
            \draw[->] (z2o) -- (v2) {}; 
            \node (labelz2ov2)[label=right:${d_G(s,v_2)}$] at (8.05,-1) {};
            \draw[->] (v3) -- (z2i) node[midway, label=right:$0$] {}; 
            \draw[->] (z2i) -- (z2o) {};
            \node (labelz2iz2o)[label=below:$0$] at (10,-1.65) {};
            \end{scope}
            \node (y0o)[label=left:$y_{0_o}$] at (1,-4)  {};
            \node (y0i)[label=below:$y_{0_i}$] at (3,-4)  {};
            \node (y1o)[label=right:$y_{1_o}$] at (5,-4)  {};
            \node (y1i)[label=right:$y_{1_i}$] at (7,-4)  {};
            \node (y2o)[label=right:$y_{2_o}$] at (9,-4)  {};
            \node (y2i)[label=right:$y_{2_i}$] at (11,-4)  {};
            \begin{scope}[->,every node/.style={ right}]
            \draw[->] (y0o) -- (z0o) node[midway, label=left:$0$] {};
            \draw[->] (y1o) -- (z1o) node[midway, label=left:$0$] {};
            \draw[->] (y2o) -- (z2o) node[midway, label=left:$0$] {};
            \draw[->, color=gray] (z0i) -- (y0i) {};
            \draw[->, color=gray] (z1i) -- (y1i) {};
            \draw[->, color=gray] (z2i) -- (y2i) {};
            \end{scope}
            \node (A)[label=below:$A$] at (5,-7) {};
            \begin{scope}[->,every node/.style={ right}]
            \draw[->] (y0o) -- (A) {};
			 \draw[->] (y1o) -- (A) {};  
   			 \draw[->] (A) -- (y0i) {};
			 \draw[->] (A) -- (y1i) {};
			 \draw[->] (y2o) -- (A) {};     
			 \draw[->] (A) -- (y2i) {};
            \end{scope}
            \node (C10)[label=below:$C_{1,0}$] at (5,-9) {};
            \node (C11)[label=below:$C_{1,1}$] at (7,-9) {};
            \node (C20)[label=below:$C_{2,0}$] at (5,-11) {};
            \node (C21)[label=below:$C_{2,1}$] at (7,-11) {};
            \begin{scope}[->,every node/.style={ right}]
            \draw[->, line width=1pt] (y0o) to  [bend right=60, distance=10mm] (C10) {};
            \draw[->, line width=1pt] (y0o) to  [bend right=60, distance=10mm] (C20) {};
            \draw[->, line width=1pt] (C11) to  [bend right=60, distance=10mm] (y0i) {};
            \draw[->, line width=1pt] (C21) to  [bend right=60, distance=10mm] (y0i) {};
            \draw[->, line width=1pt] (y1o) to  [bend left=60, distance=10mm] (C11) {};
            \draw[->, line width=1pt] (y1o) to  [bend left=60, distance=10mm] (C20) {};
            \draw[->, line width=1pt] (C10) to  [bend right=60, distance=10mm] (y1i) {};
            \draw[->, line width=1pt] (C21) to  [bend right=60, distance=10mm] (y1i) {};
            \draw[->, line width=1pt] (y2o) to  [bend left=60, distance=10mm] (C10) {};
            \draw[->, line width=1pt] (y2o) to  [bend left=60, distance=10mm] (C21) {};
            \draw[->, line width=1pt] (C11) to  [bend right=60, distance=10mm] (y2i) {};
            \draw[->, line width=1pt] (C20) to  [bend right=120, distance=20mm] (y2i) {};            
            \end{scope}
            \end{tikzpicture}}
    \caption{$G''$ for $l = 3$ in the reduction: directed 2-SiSP $\lesssimsp_{m+n}$ Betweenness Centrality. The gray and the bold edges have weight $M'$ and $\frac{1}{3}M'$ respectively. All the outgoing (incoming) edges from (to) $A$ have weight $M'+q$ ($0$).}
    \label{fig:2SiSPBetweennessCentrality}
\end{figure}

We observe that for each $0\leq j\leq l-1$, a shortest path from $y_{j_o}$ to $y_{j_i}$ with $(z_{j_i},y_{j_i})$ as the last edge has weight 
equal to $M' + d_{G''} (z_{j_o},z_{j_i})$.

\textit{(ii) We now show that the Betweenness Centrality of $A$, i.e. $BC(A)$, is equal to $l$ iff $q < d_{G''} (z_{j_o},z_{j_i})$ for each $0\leq j\leq l-1$.}
The only paths that passes through the vertex $A$ are from vertices $y_{j_o}$'s to vertices $y_{j_i}$'s.
For $j\neq k$, as noted in the proof of Lemma~\ref{lemma:2SiSPRadius}, there exists some $r, s$ such that there is a path from $y_{j_o}$
to $y_{k_i}$ that goes through $C_{r,s}$ and has weight equal to $\frac{2}{3}M'$.
However a path from $y_{j_o}$ to $y_{k_i}$ has weight $M'+q$, which is strictly greater than $\frac{2}{3}M'$ and hence the pairs 
$(y_{j_o},y_{k_i})$ does not contribute to the Betweenness Centrality of $A$.

Now if $BC(A)$, is equal to $l$, it implies that the shortest paths for all pairs $(y_{j_o},y_{j_i})$ passes through $A$ and there is exactly
one shortest path for each such pair.
Hence for each  $0\leq j\leq l-1$, $M' + q < M' + d_{G''} (z_{j_o},z_{j_i})$.
Thus  $q < d_{G''} (z_{j_o},z_{j_i})$ for each $0\leq j\leq l-1$.

On the other hand if $q < d_{G''} (z_{j_o},z_{j_i})$ for each $0\leq j\leq l-1$, then the path from $y_{j_o}$ to $y_{j_i}$ with 
$(z_{j_i},y_{j_i})$ as the last edge has weight $M' + d_{G''} (z_{j_o},z_{j_i})$.
However the path from $y_{j_o}$ to $y_{j_i}$ passing through $A$ has weight $M' + q < M' + d_{G''} (z_{j_o},z_{j_i})$.
Hence every such pair contributes $1$ to the Betweenness Centrality of $A$ and thus $BC(A) = l$.

Thus using (ii), we just need to find the minimum value of $q$ such that $BC(A) < l$ in order to compute the value 
$\min_{0\leq j\leq l-1}  d_{G''} (z_{j_o},z_{j_i})$.
We can find such $q$ by performing a binary search in the range $0$ to $nM$ and computing $BC(A)$ at every layer.
Thus we make $O(\log nM)$ calls to the Betweenness Centrality algorithm.

As observed in the proof of Lemma~\ref{lemma:ReplacementANSC}, we know that the shortest path from $z_{j_o}$ to $z_{j_i}$
corresponds to the replacement path for the edge $(v_j,v_{j+1})$ lying on $P$. 
Thus by making $O(\log nM)$ calls to the Betweenness Centrality algorithm, we can compute the second simple shortest path from $s$ to
$t$ in $G$.
This completes the proof.

The cost of this reduction is $O((m+n\log n)\cdot\log nM)$.
\vspace{-.1in}
\end{proof}

\subparagraph{II. \underline{ANSC to Pos ANBC}:}

We now describe a tilde-sparse reduction from the ANSC problem to the All Nodes Positive Betweenness Centrality problem
 (Pos ANBC).
Sparse reductions
from \MinWeightDelta{}  and from Diameter to  
 Pos ANBC
are
given in~\cite{AGW15}.
However, \MinWeightDelta{} can be solved in $O(m^{3/2})$ time, and Diameter is not known to be 
MWC-hard,
 hence  
neither
of these reductions
can be used to show hardness of the All Nodes Positive Betweenness Centrality problem relative to the MWC Conjecture.
(Recall that Pos ANBC is not known to be subcubic equivalent to APSP.)

Our reduction is similar to the reduction from 2-SiSP to the Betweenness Centrality problem, but instead of computing betweenness 
centrality through one vertex, it
computes
the positive betweenness centrality values for $n$ different nodes.
We 
first split every vertex $x$ into vertices $x_o$ and $x_i$ such that the shortest path from $x_o$ to $x_i$ corresponds 
to the shortest cycle passing through $x$ in the original graph.
We then add additional vertices $z_x$ for each vertex $x$ in the original graph and connect it to the vertices $x_o$ and $x_i$ such that 
the only shortest path passing through $z_x$ is from $x_o$ to $x_i$.
We then perform binary search on the edge weights for the edges going from $z_x$ to $x_i$ with oracle calls to the Positive Betweenness 
Centrality problem, to compute the weight of the shortest cycle passing through $x$ in the original graph.
Our reduction is described in Lemma~\ref{lemma:ANSCAllNodesPositive}.

\begin{lemma}	\label{lemma:ANSCAllNodesPositive}
In weighted directed graphs, ANSC $\lesssimsp_{m+n}$ Pos ANBC
\end{lemma}
\vspace{-.25in}
\begin{proof}
We are given an input graph $G = (V,E)$ and we wish to compute the ANSC in $G$.
Let $M$ be the largest edge weight in $G$.

\textit{(i) Constructing $G'$:}
Now we construct a graph $G'$ from $G$.
For each vertex $x\in V$, we replace $x$ by vertices $x_i$ and $x_o$ and we place a directed edge of weight $0$ from $x_i$ to $x_o$,
and we also replace each incoming edge to (outgoing edge from) $x$ with an incoming edge to $x_i$ (outgoing edge from $x_o$) in $G'$.
We can observe that the shortest path from $x_o$ to $x_i$ in $G'$ corresponds to the shortest cycle passing through $x$ in $G$.

For each vertex $x\in V$, we add an additional vertex $z_x$ in $G'$ and we add an edge of weight $0$ from $x_o$ to $z_x$ and an edge
of weight $q_x$ (where $q_x$ lies in the range from $0$ to $nM$) from $z_x$ to $x_i$.
 
Figure~\ref{fig:ANSCAllNodesPositive} depicts the full construction of $G'$ for $n= 3$.

\begin{figure}
\scriptsize
    \centering
        \makebox[.3\textwidth]{
            \begin{tikzpicture}[scale=0.5, every node/.style={circle, draw, inner sep=0pt, minimum width=5pt}]
            \node (1o)[label=above:$1_o$] at (0,0)  {};
            \node (1i)[label=above:$1_i$] at (2,0)  {};
            \node (2o)[label=above:$2_o$] at (3,0)  {};
            \node (2i)[label=above:$2_i$] at (5,0) {};
            \node (3o)[label=above:$3_o$]	at (6,0)	{};
            \node (3i)[label=above:$3_i$] at (8,0)	{};
            \node[ellipse, color=blue, fit=(1o)(1i)(2o)(2i)(3o)(3i), inner sep=3mm](G') [label=left:$G$] {};
            \begin{scope}[->,every node/.style={ right}]
            \draw[->] (1i) -- (1o) node[midway, label=above:$0$] {};
            \draw[->] (2i) -- (2o) node[midway, label=above:$0$] {};
            \draw[->] (3i) -- (3o) node[midway, label=above:$0$] {};
            \end{scope}
            \node (z1)[label=below:$z_1$] at (1,-2) {};
            \node (z2)[label=below:$z_2$] at (4,-2) {};
            \node (z3)[label=below:$z_3$] at (7,-2) {};
            \begin{scope}[->,every node/.style={ right}]
            \draw[->] (1o) -- (z1) node[midway, label=left:$0$] {};
            \draw[->] (2o) -- (z2) node[midway, label=left:$0$] {};
            \draw[->] (3o) -- (z3) node[midway, label=left:$0$] {};
            \draw[->] (z1) -- (1i) node[midway, label=right:$q_1$] {};
            \draw[->] (z2) -- (2i) node[midway, label=right:$q_2$] {};
            \draw[->] (z3) -- (3i) node[midway, label=right:$q_3$] {};
            \end{scope}
            \end{tikzpicture}}
    \caption{$G'$ for $n = 3$ in the reduction: directed ANSC $\lesssimsp_{m+n}$ Pos ANBC.}
    \label{fig:ANSCAllNodesPositive}
\end{figure}
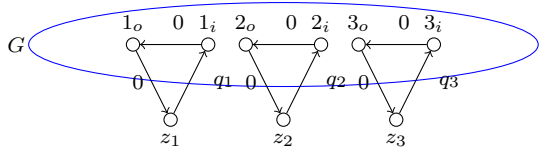

We observe that the shortest path from $x_o$ to $x_i$ for some vertex $x \in V$ passes through $z_x$ only if the shortest cycle passing
through $x$ in $G$ has weight greater than $q_x$.

\textit{(ii) We now show that for each vertex $x\in V$, Positive Betweenness Centrality of $z_x$
is true,
i.e., $BC(z_x) > 0$ iff the shortest cycle
passing through $x$ has weight greater than $q_x$.}
It is easy to see that the only path that pass through vertex $z_x$ is from $x_o$ to $x_i$ (as the only outgoing edge from $x_i$ is to 
$x_o$ and the only incoming edge to $x_o$ is from $x_i$).

Now if $BC(z_x) > 0$, it implies that the shortest path from $x_o$ to $x_i$ passes through $z_x$ and hence the path from $x_o$ to $x_i$
corresponding to the shortest cycle passing through $x$ has weight greater than $q_x$.

On the other hand, if the shortest cycle passing through $x$ has weight greater than $q_x$, then the shortest path from $x_o$ to $x_i$ 
passes through $z_x$.
And hence $BC(z_x) > 0$.

Then using $(ii)$, we just need to find the maximum value of $q_x$ such that $BC(z_x) > 0$ in order to compute the weight of the 
shortest cycle passing through $x$ in the original graph.
We can find such $q_x$ by performing a binary search in the range $0$ to $nM$ and computing Positive Betweenness Centrality for all
nodes at every layer.
Thus we make $O(\log nM)$ calls to the 
Pos ANBC algorithm.
This completes the proof.

The cost of this reduction is $O((m+n)\cdot\log nM)$.
\end{proof}

\appendix

\section{Appendix} \label{sec:appendix}

\subsection{Definitions of Graph Problems} \label{sec:definitions}

\subparagraph*{All Pairs Shortest Distances (APSD).}
Given a graph $G = (V,E)$, the APSD problem is to compute the shortest path distances for every pair of vertices in $G$.

\vspace{-.2in}

\subparagraph*{APSP.}

This is the problem of
computing the shortest path distances for every pair of vertices in $G$
 together with a
concise
 representation of the shortest paths, which in our case is an $n\times n$ matrix, $Last_G$, 
that contains, 
 in position $(x,y)$,
 the predecessor vertex of $y$ on a shortest path from 
$x$ to $y$.

 \vspace{-.2in}

\subparagraph*{Minimum Weight Cycle (MWC).}
Given a graph $G = (V,E)$, the minimum weight cycle problem is to find the weight of a minimum weight cycle in $G$.

\vspace{-.2in}

\subparagraph*{All Nodes Shortest Cycles (ANSC).}
Given a  graph $G = (V,E)$, the ANSC problem is to find the weight of a shortest cycle through each vertex in $G$.

\vspace{-.2in}

\subparagraph*{Replacement Paths.}
Given a graph $G = (V,E)$ and a pair of vertices $s, t$, the replacement paths problem is to find, for each edge $e$ lying on the shortest
path from $s$ to $t$, a shortest path from $s$ to $t$ avoiding the edge $e$.

\vspace{-.2in}

\subparagraph*{$k$-SiSP.}
Given a graph $G = (V,E)$ and a pair of vertices $s, t$, the $k$-SiSP problem is to find the $k$ shortest simple paths from $s$ to $t$: 
the $i$-th path
must be
different from 
first
$(i-1)$ paths and  
must have
weight greater than or equal to the
weight of any of these $(i-1)$ paths.

\vspace{-.2in}

\subparagraph*{$k$-SiSC.}
The corresponding cycle version of $k$-SiSP is known as $k$-SiSC, where the goal is to compute the $k$ shortest simple cycles through a
given vertex $x$, such that the $i$-th cycle generated is different from all previously generated $(i-1)$ cycles and has weight greater than 
or equal to the weight of any of these $(i-1)$ cycles. 

\vspace{-.2in}

\subparagraph*{Radius.}

For a given graph $G = (V,E)$, the Radius problem is to compute the value $\min_{x\in V}$ $\max_{y\in V}$ $d_G (x,y)$.
The {\it center} of a graph is the vertex $x$ which minimizes this value.

\vspace{-.2in}

\subparagraph*{Diameter.}

For a given graph $G = (V,E)$, the Diameter problem is to compute the value $\max_{x,y\in V}$ $d_G (x,y)$.

\vspace{-.2in}

\subparagraph*{Eccentricities.}

For a given graph $G = (V,E)$, the Eccentricities problem is to compute the value $\max_{y \in V} d_G (x,y)$ for each vertex $x \in V$.

 \vspace{-.2in}

\subparagraph*{Betweenness Centrality (BC).}
For a given graph $G = (V,E)$ and a node $v\in V$, the Betweenness Centrality of $v$,
$BC(v)$,
 is the value $\sum_{s,t \in V, s,t\neq v} \frac{\sigma_{s,t} (v)}{\sigma_{s,t}}$, where $\sigma_{s,t}$ is the number of shortest paths from $s$ to $t$ and $\sigma_{s,t} (v)$ is
the number of shortest paths from $s$ to $t$ passing through $v$.

As in~\cite{AGW15} we assume that the graph has unique shortest paths, hence 
$BC(v)$
is simply the number of $s,t$ pairs
such that the shortest path from $s$ to $t$ passes through $v$.

\vspace{-.2in}

\subparagraph*{All Nodes Betweenness Centrality (ANBC).}

The all-nodes version of Betweenness Centrality: 
determine $BC(v)$ for all vertices.

\vspace{-.2in}

\subparagraph*{Positive Betweenness Centrality (Pos BC).}

Given a graph $G=(V,E)$ and a vertex $v$, the Pos BC problem is to deternine if $BC(v) > 0$.

\vspace{-.2in}

\subparagraph*{All Nodes Positive Betweenness Centrality (Pos ANBC).}
 
The all-nodes version of Positive Betweenness Centrality.

\vspace{-.2in}

\subparagraph*{Reach Centrality (RC).}

For a given graph $G = (V,E)$ and a node $v \in V$, the Reach Centrality of $v$,
$RC(v)$,
 is the value $\max_{s,t \in V: \\ d_G (s,v) + d_G (v,t) = d_G (s,t)} {\min (d_G (s,b), d_G (b,t))}$.

\subsection{$k$-SiSC Algorithm : Undirected Graphs} \label{sec:kSiSC}

This section deals with an application of our {\it bit-sampling technique} to obtaining a new near-linear time algorithm for $k$-SiSC in
undirected graphs (see definition below).
Note that this problem is not in the $mn$ class and this result is included here as an application of the {\it bit-sampling technique}.

$k$-SiSC is the problem of finding $k$ simple shortest cycles passing through a vertex $v$.
Here the output is a sequence of $k$ simple cycles through $v$ in non-decreasing order of weights such that the $i$-th cycle in the output is 
different from the previous $i-1$ cycles.
The corresponding path version of this problem is known as $k$-SiSP and is solvable in near linear time in undirected graphs~\cite{KIM82}.

We now use our {\it bit-sampling technique} (described in Section~\ref{sec:undir}) to get a near-linear time algorithm for $k$-SiSC, which was
not previously known.
We obtain this $k$-SiSC algorithm by giving  
 a tilde-sparse $\tilde O(m+n)$ time
reduction from $k$-SiSC to $k$-SiSP.
This reduction uses our {\it bit-sampling technique} for sampling the edges incident to $v$ and creates
 $\lceil \log n \rceil$ different graphs. 
 Here we only use index $i$ of our bit-sampling method.

\begin{lemma}	\label{lemma:unkSiSCkSiSP}
In undirected graphs, $k$-SiSC $\lesssimsp_{(m+n)}$ $k$-SiSP.
\end{lemma}

\vspace{-.1in}
\begin{proof}
Let the input be $G = (V,E)$ and let $x \in V$ be the vertex for which we need to compute $k$-SiSC.
Let $\mathcal{N}(x)$ be the neighbor-set of $x$.
We create $\lceil \log n \rceil$ graphs $G_i = (V_i, E_i)$ such that $\forall 1 \leq i\leq \lceil \log n\rceil$, $G_i$ contains two additional
vertices $x_{0,i}$ and $x_{1,i}$ (instead of the vertex $x$) and $\forall y \in \mathcal{N}(x)$, the edge $(y,x_{0,i}) \in E_i$ if $y$'s 
$i$-th bit is $0$, otherwise the edge $(y,x_{1,i}) \in E_i$. 
This is our bit-sampling method.

The construction takes $O((m+n) \cdot \log n)$ time and we observe that every cycle through $x$ will appear as a path from $x_{0,i}$ to $x_{1,i}$ in
at least one of the $G_i$.
Hence, the $k$-th shortest path in the collection of $k$-SiSPs from $x_{0,i}$ to  $x_{1,i}$ in $\log n$  $G_i$,  $1 \leq i \leq \lceil \log n \rceil$
(after removing duplicates), corresponds to the $k$-th SiSC passing through $x$. 
\end{proof}

Using the undirected $k$-SiSP algorithm in~\cite{KIM82} that runs in $O(k\cdot (m + n\log n))$, we obtain 
an $O(k\log n\cdot (m + n\log n))$ time algorithm for  $k$-SiSC in undirected graphs.

\bibliographystyle{abbrv}
\bibliography{references-new}

\end{document}